\documentclass{article}
\usepackage{subcaption}
\usepackage{url}
\usepackage{bm}
\usepackage{hyperref}

\usepackage{amsmath}
\usepackage{svg}
\usepackage{tikz}
\usepackage{amssymb}
\usepackage{titlesec}
\usepackage{chngcntr}
\usepackage{adjustbox}
\usepackage{quantikz}
\usepackage{amsthm}
\usepackage{thm-restate}
\usepackage{booktabs}
\usepackage{makecell}

\titleclass{\subsubsubsection}{straight}[\subsubsection]
\newcounter{subsubsubsection}[subsubsection]
\renewcommand\thesubsubsubsection{\thesubsubsection.\arabic{subsubsubsection}}
\titleformat{\subsubsubsection}[runin]
  {\normalfont\normalsize\bfseries}{\thesubsubsubsection}{1em}{}

\usepackage{algorithm}
\usepackage{algpseudocode}
\usepackage{authblk}

\usepackage{pdflscape}
\usepackage{rotating}
\usepackage{soul}
\usepackage{authblk}

\title{A Quantum Approach to Stochastic Optimization in Insurance Underwriting}
\author[1]{Mitchell Bordelon}
\author[1]{Maurice Garfinkel}
\author[2]{Vivek Dixit}
\author[1]{Thomas Whitehead}
\author[1]{Jenny Holzbauer}
\author[2]{Guillermo Mijares Vilarino}
\author[2]{Alberto Maldonado Romo}
\author[2]{Abhijit Mitra}
\author[2]{Vaibhaw Kumar\thanks{vaibhaw.kumar@ibm.com}}
\author[1]{Jean Utke\thanks{jutke@allstate.com}}

\affil[1]{Allstate}
\affil[2]{IBM Quantum}

\newcommand{\later}[1]{}

\newcommand{\kett}[1]{| #1\rangle}

\DeclareMathAlphabet{\mathpzc}{OT1}{pzc}{m}{it}
\newcommand{\XX}{X\!X}
\newcommand{\XPXX}{X\!\!+\!\!\XX}

\usepackage{tikz}
\usetikzlibrary{quantikz2}
\tikzset{
    operator/.style={draw, fill=blue!10, minimum size=1.5em},
    phase/.style={draw=none, fill=none, minimum size=1.5em},
}

\begin{document}

\maketitle 
\begin{abstract}
The presence of stochastic elements in combinatorial optimization problems makes them particularly challenging, as such problems quickly become intractable for classical computers even at relatively small sizes. In this work, we propose a novel quantum-classical hybrid scheme for solving a class of stochastic optimization problems known as chance-constrained knapsack problems, in which item weights follow probability distributions and constraints may be violated within a specified risk tolerance. Our method employs knapsack-specific QAOA-based circuits to generate samples which, when combined with a new self-consistent classical recovery scheme introduced in this work, produce high-quality solutions. Experiments carried out on IBM Heron processors, using circuits with depths up to 177 and comprising 3443 gates acting on as many as 150 qubits, yield solutions that indicate performance comparable to classical optimization schemes. The proposed quantum-classical scheme paves the way to tackling such problems, with the potential to outperform approaches that rely solely on classical computation.\footnote{The source code will be made publicly available upon acceptance of the manuscript.}
\end{abstract}

\section{Introduction}

Intrinsic stochasticity in combinatorial optimization problems introduces uncertainty that substantially increases computational complexity and limits the effectiveness of deterministic solution methods. This uncertainty naturally arises in applications including finance, logistics, insurance, and other risk-reward business problems, where fluctuations in demand, prices, or system disruptions require risk-based decision-making. In these settings, uncertainty is typically modeled through chance-constrained formulations, in which constraints are enforced only probabilistically and violations are allowed with a small probability $\epsilon$ that reflects a specified level of risk tolerance. Decision making under such risk-reward trade-offs amounts to balancing probabilistic feasibility against optimization objectives.

In particular, decision tasks in insurance underwriting naturally map to chance-constrained knapsack (CCKP) problems. By abstracting away operational domain-specific details, CCKP formulations capture the core insurance objective weighing expected return against acceptable risk, while preserving a defining feature of the problem: the stochastic nature of constraints. In real-world scenarios, this stochasticity stems from predictive models used to assess risk and estimate uncertainty. Similar characteristics are shared by many other business problems involving discrete decision-making under uncertainty, making the CCKP formulation widely applicable. 

However, despite their practical relevance, chance-constrained optimization problems pose significant computational challenges. They often contain feasible regions that are non-convex, high-dimensional, and overall difficult to characterize, complicating their analysis and practically finding solutions. Classical approaches typically seek to transform probabilistic constraints into deterministic equivalents or tractable approximations \cite{GOYAL2010161}. However, traditional methods, including mixed-integer programming formulations \cite{rossi2025mixedintegerlinearprogrammingapproximations} and heuristic approximations such as scenario-based approaches \cite{Zeng2018-vf} and Monte Carlo simulations \cite{Morton1998}, often scale poorly with problem size. Notably, accurately modeling uncertainty distributions, evaluating constraint violations, and exploring the resulting vast solution spaces leads to prohibitive runtimes or suboptimal solutions. 

In parallel, the application of quantum algorithms to knapsack and other combinatorial optimization problems has recently attracted significant attention~\cite{Mohseni2024challenges,Sharma2024,Bozejko2024,Herman2023,Hegade2024}. Notable examples include Copula-QAOA, which employs warm-started mixers that are biased toward greedy solutions~\cite{vanDam2021}, and the Quantum Tree Generator, which facilitates amplitude amplification across all feasible solutions~\cite{Wilkening2025}. In addition, recent utility-scale experiments have showcased constrained optimization on systems with up to 150 qubits using shallow mixer circuits~\cite{Mohseni2026}. However, to the best of our knowledge, all of these existing gate-based approaches exclusively address problems with deterministic constraints. On the stochastic side, quantum annealing has been applied to chance-constrained problems through scenario-based MILP reformulations~\cite{Ribes2025}, demonstrating the viability of quantum approaches for optimization under uncertainty while still relying on classical scenario sampling.

This work introduces a novel quantum-classical framework that extends QAOA-based approaches to CCKPs, leveraging gate-based quantum hardware to tackle stochastic optimization problems. 
The proposed quantum method employs a variational circuit that embeds chance constraints through a deterministic reformulation, enabling efficient generation of a pool of high-quality bitstrings that respect a prescribed risk level and exhibit an increased likelihood of maximizing knapsack value. At its core, the circuit architecture takes inspiration from the QAOA-based Lagrangian Duality Discretized Adiabatic (LD-DA) circuit originally proposed for deterministic knapsack problems~\cite{53gd-2374}, but it is augmented with a variational component for this stochastic setting. 

Additionally, to mitigate known challenges in variational training, such as barren plateaus~\cite{Larocca_2025, mcclean2018barren}, a parameter transfer strategy based on constraint-alignment is introduced. This enables parameters to be transferred from relatively small problem instances to larger ones, where the exponential growth of the solution space renders direct optimization challenging. These parameters are subsequently refined with a limited number of fine-tuning optimization iterations to obtain samples characterized by sufficiently low objective values, which are then post-processed by a novel self-consistent classical recovery scheme customized for CCKP problems. 

Furthermore, large knapsack instances pose substantial challenges for classical solvers. Accordingly, this work includes classical benchmarking on increasingly large CCKPs to assess the onset of computational intractability, motivate the proposed quantum-classical approach, and contextualize performance comparisons. Classical benchmarking experiments obtained using Gurobi were used to compare against quantum hardware experiments conducted on IBM Heron processors, involving problem sizes ranging from 20 to 150 items. This corresponds to circuit depths up to 177 comprising a total of 3443 gates. Overall, the quantum solutions are comparable in quality to classical heuristic methods for all tested instances 
under the same evaluation criteria.

The remainder of this paper is organized as follows: \autoref{sec:methods} introduces the chance-constrained knapsack formulation, the variational quantum circuit construction, the hybrid training, and the post-processing methodology. \autoref{sec:results} presents numerical results from both quantum simulations and hardware experiments alongside classical benchmarks followed by conclusions and outlook in \autoref{sec:conclusions}. The appendices (\autoref{app:derivation}--\autoref{app:classical_solvers}) provide supporting theoretical proofs, extended analytical derivations, details on mixer behavior and circuit geometry, pseudocode for post-processing and recovery algorithms, and additional experimental and implementation details that supplement the main text.

\section{Methods}\label{sec:methods}

\subsection{Chance-constrained knapsack problem}
The classical (deterministic) multi-dimensional knapsack problem extends the standard 0--1 knapsack by imposing multiple constraints. The chance‑constrained knapsack problem (CCKP) generalizes this formulation by replacing the deterministic item weights with random variables, which accounts for uncertainty in resource consumption. Specifically, for item $i \in \{1,\ldots,n\}$ and constraint $m \in \{1,\ldots,M\}$, the random weight $\tilde{w}_i^{(m)}$ is drawn from a distribution family $\Delta(\mu,\sigma)$ with mean $\mu_i^{(m)}$ and standard deviation $\sigma_i^{(m)}$.

The CCKP seeks a binary decision vector $x \in \{0,1\}^n$, where $x_i=1$ indicates item $i$ is selected. An optimal $x$ maximizes total value while ensuring constraints on weight capacity are satisfied with high probability. Formally, the problem is defined as:

\begin{equation}
\label{eq:chance-constrained}
\begin{aligned}
\max_{\bm x\in\{0,1\}^n}\;&\sum_{i=1}^n v_i x_i \\
\text{s.t. }&\mathbb{P}\!\left(\sum_{i=1}^n \tilde{w}_i^{(m)} x_i \leq C^{(m)}\right) \geq 1-\epsilon^{(m)}, \quad m \in \{1,2,\dots, M\}. 
\end{aligned}
\end{equation}

\noindent For certain weight distributions $\Delta$, the CCKP can be reformulated with deterministic constraints as follows:
\begin{restatable}[]{theorem}{knapsacktheorem}
\label{theorem:1}
For the mutually independent knapsack weights $\tilde{w}_i^{(m)}$  chosen from a distributions $\Delta(\mu,\sigma)$ closed under convolution with mean $\mu_i^{(m)}$ and standard deviation $\sigma_i^{(m)}$, and its cumulative distribution function $\Phi$ strictly increasing and invertible, the chance-constrained problem in \eqref{eq:chance-constrained} can be rewritten as

\begin{align}
\quad
\max_{\bm x \in\{0,1\}^n} \quad & \sum_{i=1}^n v_i x_i \\
\text{s.t.} \quad & \sum_{i=1}^n \mu_i^{(m)} x_i +
                    \Phi^{-1}(1-\epsilon^{(m)})\sqrt{\sum_{i=1}^n {\sigma_i^{(m)}}^2 x_i}
                    \leq C^{(m)} \label{eq:socp-constraint} \quad \\
                    &\forall \quad m \in \{1,2,\dots, M\} \notag
\end{align}
where $\Phi^{-1}$ is the inverse of the cumulative distribution function (CDF) of the distribution $\Delta(0,1)$ and some small $\epsilon^{(m)} > 0$.\footnote{For distributions satisfying the stated conditions, the probabilistic constraint in~\eqref{eq:chance-constrained} admits an exact reformulation as the second‑order cone constraint in~\eqref{eq:socp-constraint}. Classical solvers such as Gurobi handle these constraints by squaring both sides to obtain an equivalent quadratic constraint, enabling the use of second-order cone programming (SOCP)~\cite{GOYAL2010161}.}
\end{restatable}

\autoref{app:derivation} and \cite{GOYAL2010161} provide the proof of this result. As stated in Theorem \ref{theorem:1}, the distribution $\Delta$ must be closed under convolution, with a CDF that is strictly increasing and invertible. Several commonly used distributions belonging to the family of $\alpha$-stable distributions including Normal, Cauchy, and, L\'evy
satisfy these properties. In this paper, we restrict our attention to item weights drawn from Normal distributions. 

\subsection{Classical problem instance creation}\label{sec:instances}

Classical benchmarks are constructed using a deterministic mixed‑integer second\-order cone program (MISOCP) reformulation of the CCKP in \eqref{eq:socp-constraint} by squaring both sides. This representation isolates the interaction between combinatorial decision variables and quadratic second‑order cone constraints, enabling systematic exploration of regimes that exhibit significant classical computational difficulty. Hard problem instances are generated and solved using Gurobi 12.0.2~\cite{gurobi} on an Intel Xeon Silver 4310 CPU (2.10 GHz) with 12 cores. Problem instances are chosen to lie within regimes accessible to both classical and quantum hardware, specifically from combinations of $n \in \{20, 40, 75, 100, 150\}$ and $M \in \{1, 5, 20, 50, 100\}$ in conjunction with the following additional parameters:

\begin{enumerate}
    \item Maximum item value $v_{\max} \in \{4, 20, 50, 100\}$ as the upper bound for uniform sampling of item values $v_i \in (0, v_{\max}]$.
    \item Maximum item weight mean $\mu_{\max} \in \{5, 25, 50, 100\}$, where individual item weight means $\mu_i^{(m)}$ are sampled uniformly from the interval $(0, \mu_{\max}]$.
    \item Item weight standard deviation $\sigma$, where for all items, the variance $\sigma^2 \in \{0.5, 1, 1.5, 2, 3\}$. 
   \item Approximate constraint density ratio $r_D \in \{0.25, 0.5, 0.75\}$, where $r_D \approx \frac{\sum_{i=1}^n \mathbf{1}\bigl(\mu_i^{(m)} \neq 0\bigr)}{n}$, and this ratio informs a random subset selection of $\mu_i^{(m)}$ to be set to zero, together with the corresponding subset of $\sigma_i^{(m)}$.
    \item Approximate capacity ratio $r_C \in \{0.25, 0.5, 0.75\}$, where $r_C \approx \frac{\sum_{i=1}^n \mu_i^{(m)}}{C^{(m)}}$, in order to determine $C^{(m)}$ after the parameters $\mu_i^{(m)}$ have been specified according to the rules outlined above.
    \item Fixed value for $\epsilon^{(m)}=0.2$ for all $m$ and henceforth denote $\epsilon^{(m)} \equiv \epsilon$. This permits comparisons in algorithm performance across instances rather than changes in risk tolerance.
\end{enumerate}

In the first optimization phase, a comprehensive parameter sweep is performed over multiple dimensions characterizing the stochastic knapsack problem, resulting in a total of 3,900 instances. Each instance is allotted a maximum runtime of three minutes, and instances which Gurobi could prove optimality within this time limit are labeled as “easy.” Empirical results (see \autoref{fig: gurobi}) indicate that problem difficulty is most strongly affected by increasing values of $n$ followed by $M$, with additional sensitivity to smaller values of $\mu_{\max}$. Larger values of $\sigma^2$ further increase runtime and exacerbate optimality gaps, whereas $v_{\max}$, $r_C$, and $r_D$ have negligible impact. Empirically, the parameter configuration $v_{\max}=50$, $\mu_{\max}=5$, $r_C=0.25$, and $r_D=0.5$ consistently yield challenging instances as $n$ and $M$ increased.  Accordingly, these parameter values are fixed for the following two additional phases.

The second optimization phase conducts a limited exploration of the weight variance $\sigma^2$. In the third phase, $\sigma^2$ is fixed to $\sigma^2=1$ while investigating solver performance as the number of items scales.  Sixteen instances are created by varying $n \in \{20, 40, 50, 75, 100, 150, 500, 1000\}$ and $M \in \{5, 20\}$. The problem instances generated in this step are used to assess the performance of multiple classical solvers, as summarized in~\autoref{fig: classical benchmarking}. From this set, only instances with $n \leq 500$ are retained as candidates for quantum simulation and hardware execution, forming the set of problems addressed by the quantum-classical algorithm experiments below.

\begin{figure}
    \centering
    \includegraphics[width=1.0\linewidth]{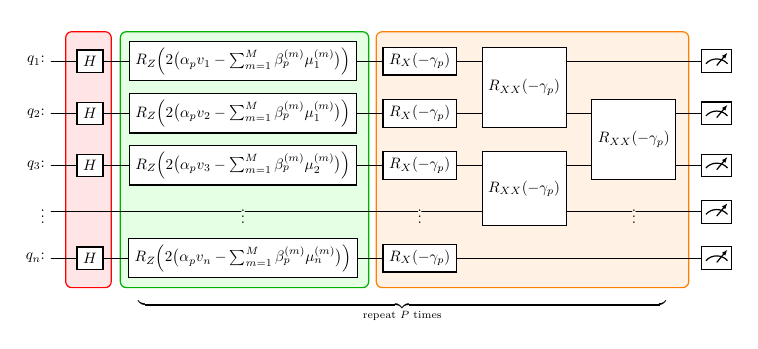}
    \caption{QAOA-based variational quantum circuit with $P$ layers and three variational parameters $\boldsymbol{\theta}=(\boldsymbol{\alpha, \beta, \gamma})\in I\!\!R^{P\times (2+M)}$. The circuit contains parametrized single qubit $X$ and two-qubit $\XX$ rotations where two-qubit gates are arranged over a line. The phase separation block applies diagonal $R_Z(2(\alpha_p v_i- \sum_{m=1}^M \beta_p^{(m)} \mu_i^{(m)}))$ gates, where layers are indexed by $p=1,\ldots,P$.}
    \label{fig:ansatz}
\end{figure}

\subsection{Quantum circuit construction}

%

To evaluate these instances on quantum hardware, a circuit ansatz is required that mirrors the structure of the underlying optimization problem while remaining compatible with gate‑based devices. Accordingly, the proposed quantum-classical algorithm employs a QAOA-based~\cite{farhi2014quantum} variational ansatz with $Z$, $X$ and $\XX$ rotations, as shown in \autoref{fig:ansatz}. This construction is inspired by the Trotterized time-evolution operator arising from the Lagrangian Dual-Discretized Adiabatic Quantum Computation (LD-DAQC) scheme that has exhibited improved performance relative to canonical QUBO-based circuits on a large set of knapsack problems~\cite{53gd-2374}. The ansatz incorporates knapsack item values and mean weights into $Z$-gate rotations with the mixer as single $X$ and two-qubit $\XX$ gates. 

Formally, the linear geometry ansatz can be split into two Hamiltonians, $H_{\mathrm{problem}}$ and $H_\mathrm{{mixer}}$, defined as

\begin{equation}
    H_{\mathrm{problem}}^{(p)}
    = \sum_{i=1}^{n}
    \left(\alpha_p v_i - \sum_{m=1}^{M}\beta_p^{(m)}\mu_i^{(m)}\right) Z_i,
    \label{eq:hproblem}
\end{equation}

\begin{equation}
    H_{\mathrm{mixer}}^{(p)}
    = -\frac{1}{2}\sum_{i=1}^{n} (\gamma_p)X_i
    - \frac{1}{2}\sum_{i=1}^{n-1} (\gamma_p)X_i X_{i+1},
    \qquad p \in \{1,\ldots,P\}
    \label{eq:hmixer}
\end{equation}
for variational parameters $\alpha, \beta, \gamma$ and each item $i$, constraint $m$, and QAOA layer $p$. The linear two-qubit entanglement map requires relatively low qubit connectivity and makes it well-suited for IBM Quantum processors with heavy-hex and square lattice qubit topology.

\subsubsection{Toy example}
\label{sec:toy-example}
While the full circuit structure is illustrated in \autoref{fig:ansatz}, its effect on constrained knapsack dynamics is most easily understood in a minimal setting. To build intuition, the role of $Z$-diagonal phase separator $H_{\mathrm{problem}}^{(p)}$  and the $\XPXX$ mixer terms in $H_{\mathrm{mixer}}^{(p)}$ are first analyzed through a minimal two-state toy example using a minimal two‑state toy example with a single constraint and linear geometry ($M=1$).\footnote{We omit in this section the superscripts $^{(m)}$.} 

\paragraph{Two-state subspace.} Consider a knapsack of $n$ items arranged on a cycle with values $v_j$ and weights $\mu_j$ for $j = 1,\dots,n$.
Focusing on two adjacent items $j$ and $j{+}1$ on the mixer graph, we examine two computational-basis states $\kett{a}=\kett{1_j0_{j+1}}$ and $\kett{b}=\kett{0_j1_{j+1}}$ that differ only on these qubits while all other selections remain identical. The states $\kett{a}$ and $\kett{b}$ are connected by an $\XX$ swap ($\lvert 10\rangle \leftrightarrow \lvert 01\rangle$) along the edge $(j,\, j{+}1)$. Restricting to this two-state subspace, the uniform superposition induced by the initial state $\lvert+\rangle^{\otimes n}$ takes the form
\begin{equation}
    \lvert\psi_0\rangle = \tfrac{1}{\sqrt{2}}\bigl(\lvert a\rangle + \lvert b\rangle\bigr).
\end{equation}
Acting on $\lvert\psi_0\rangle$, the $Z$-diagonal phase separator imparts opposing phases on the two configurations, yielding
\begin{equation}
    \lvert\psi_1\rangle = \tfrac{1}{\sqrt{2}}\bigl(e^{-i\delta/2}\lvert a\rangle + e^{+i\delta/2}\lvert b\rangle\bigr),
\end{equation}

\noindent where the \emph{phase gap},
\begin{equation}
\label{eq:phase-gap}
    \delta \;\equiv\; E_a - E_b \;=\; 2\Bigl[\alpha\bigl(v_{j+1} - v_j\bigr) \;-\; \beta\bigl(\mu_{j+1} - \mu_j\bigr)\Bigr],
\end{equation} 
is determined by the difference in value and weight contributions between the two states.

When restricted to the $\{\lvert a\rangle,\lvert b\rangle\}$ subspace, the $\XX$ interaction acting on edge $(j,\,j{+}1)$ reduces to a $-\sigma_x$ operator (negative sign due to Eq.~\eqref{eq:hmixer}), as it swaps $\lvert 10\rangle \leftrightarrow \lvert 01\rangle$. Therefore, the mixer unitary on this subspace becomes 
\begin{equation}
    e^{i\gamma\sigma_x} = \cos\gamma\, I + i\sin\gamma\,\sigma_x. 
\end{equation}
Applied to $\lvert\psi_1\rangle$, the mixer produces the final state
\begin{equation}
    \lvert\psi_2\rangle = \tfrac{1}{\sqrt{2}}
    \begin{pmatrix}
        \cos\gamma\, e^{-i\delta/2} + i\sin\gamma\, e^{+i\delta/2} \\
        \cos\gamma\, e^{+i\delta/2} + i\sin\gamma\, e^{-i\delta/2}
    \end{pmatrix}.
\end{equation}
As a result, the measurement probabilities take the form
\begin{equation}
\begin{aligned}
    \mathbb{P}(a) &= \tfrac{1}{2}\lvert\cos\gamma\, e^{-i\delta/2} + i\sin\gamma\, e^{+i\delta/2}\rvert^2
          = \tfrac{1}{2}\bigl(1 - \sin 2\gamma\;\sin\delta\bigr), \\
    \mathbb{P}(b) &= \tfrac{1}{2}\bigl(1 + \sin 2\gamma\;\sin\delta\bigr),
\end{aligned}
\label{eq:measurement prob}
\end{equation}
where the population imbalance between the two configurations is given by
\begin{equation}
\label{eq:prob-transfer}
    \mathbb{P}(b) - \mathbb{P}(a) \;=\; \sin 2\gamma \;\sin\delta,
\end{equation}

When returning to the full problem, feasibility constraints reemerge: the capacity $C$ enters implicitly through a penalty term in the classical loss function rather than through the quantum circuit itself. As a result, feasibility requirements constrain how the classical optimizer selects the variational parameters~$\alpha$, $\beta$, and $\gamma$, linking the capacity to probability transfer. To make this relationship explicit, we now turn to the expected weight distribution, written as
\begin{equation}
\label{eq:expected-weight}
    \langle W \rangle = \mathbb{P}(a)\, W_a + \mathbb{P}(b)\, W_b,
\end{equation}
where $W_a = W_{\mathrm{rest}} + \mu_j$ and $W_b = W_{\mathrm{rest}} + \mu_{j+1}$ are the total weights of the two item selections and  $W_{\mathrm{rest}} = \sum_{\ell \in S \setminus\{j\}} \mu_\ell$ is the weight of the items common to both states. Substituting measurement probabilities from Eq.~\ref{eq:measurement prob}, we get:
\begin{equation}
    \langle W \rangle = W_{\mathrm{rest}} + \frac{\mu_j + \mu_{j+1}}{2} + \frac{\mu_{j+1} - \mu_j}{2}\,\sin 2\gamma\;\sin\delta.
\end{equation}
Additionally, the \emph{marginal residual}, defined as
\begin{equation}
\label{eq:marginal-residual}
    \hat{r} \;=\; \bigg(W_{\mathrm{rest}} + \frac{\mu_j + \mu_{j+1}}{2}\bigg) - C,
\end{equation}
measures the constraint violation after the fixed items and the average weight of the two swappable items.
The feasibility condition $\langle W \rangle \leq C$, equivalent to $\hat{r}\leq 0$, therefore enforces a direct constraint on the allowable probability transfer between the two configurations:
\begin{equation}
\label{eq:capacity-bound}
    \sin 2\gamma\;\sin\delta \;\leq\; \frac{2\hat{r}}{\mu_{j} - \mu_{j+1}}
    \qquad (\text{when}\; \mu_{j} < \mu_{j+1}).
\end{equation}
When $\mu_j < \mu_{j+1}$, the right-hand side of~\eqref{eq:capacity-bound} is decreasing in $\hat r$, implying that larger marginal residuals increasingly restrict feasible probability transfer, eventually forcing a reversal of flow. 

Under the same distribution, the expected value is
\begin{equation}
    \langle V \rangle = \frac{V_a + V_b}{2} + \frac{v_{j+1} - v_j}{2}\,\sin 2\gamma\;\sin\delta,
\end{equation}
which monotonically increases in $\sin 2\gamma\;\sin\delta$ when $v_{j+1} > v_j$.
The optimizer can maximize value by pushing the amplitude transfer to the capacity limit and its optimal operating point saturates the bound
\begin{equation}
\label{eq:optimal-transfer}
    \sin 2\gamma^{\star}\;\sin\delta^{\star} \;=\; \min\!\left(1,\;\frac{2\hat{r}}{\mu_{j} - \mu_{j+1}}\right).
\end{equation}

When $\hat{r} \leq (\mu_{j} - \mu_{j+1})/2$ the minimum saturates at~$1$ and the optimizer achieves full transfer unconstrained by capacity.
When $(\mu_{j} - \mu_{j+1})/2 < \hat{r}  \leq 0$, the bound is active and capacity throttles the transfer -- the optimizer must reduce~$\gamma$ or increase~$\beta/\alpha$ to stay feasible.
When $\hat{r} >  0$, the average weight already exceeds capacity, feasibility requires $\sin 2\gamma\,\sin\delta < 0$, and the optimizer reverses the amplitude flow toward the lighter item even at the cost of value.


The three $\hat{r}$-based regimes of the toy model and the optimal operating point~\eqref{eq:optimal-transfer} carry over to the chance-constrained variant with one modification: $\hat{r}$ is replaced by the \emph{constraint residual}
\begin{equation}
\label{eq:constraint-residual}
    r(\bm{x}) \;=\; L(\bm{x})- C,
\end{equation}
where $L(\bm{x})$ is the left-hand side of (\ref{eq:socp-constraint}).
Because the variance term $\sqrt{\sum \sigma_j^2 x_j}$ 
in $L(\bm{x})$ is not separable across qubits, $r$ cannot be encoded in per-item $Z$-rotations. 
This gap is bridged by the loss function~\eqref{eq:loss}, which evaluates $r$ on each sample and feeds the violation signal back to the optimizer. Over successive iterations, this drives $\beta^{\star}/\alpha^{\star}$ above the deterministic Lagrange multiplier, compensating for the variance risk that the circuit cannot represent directly, until the probability transfer on each boundary edge matches the operating point~\eqref{eq:optimal-transfer} with $\hat{r} \to r$.

\subsubsection{The role of \texorpdfstring{$H_\mathrm{{mixer}}$}{H-mixer}}

The effectiveness of the ansatz is ultimately dictated by the behavior of the $H_\mathrm{{mixer}}$ with $\XPXX$ components in Eq.~\eqref{eq:hmixer}. Its role can be evaluated by continuing to examine the simplified CCKP instance with a single constraint ($M=1$) and homogeneous weight parameters $\mu_i = \mu$, $\sigma_i = \sigma$ for all items. Substituting into the deterministic reformulation of Theorem~\ref{theorem:1}, the chance constraint becomes

\begin{equation}
\label{eq:equal-weight-constraint}
    \mu\, k + z_\epsilon\, \sigma \sqrt{k} \;\leq\; C,
\end{equation}
where $k = \sum_{i=1}^n x_i$ is the number of selected items and $z_\epsilon = \Phi^{-1}(1-\epsilon)$.
Substituting $u = \sqrt{k}$ yields the quadratic $\mu\, u^2 + z_\epsilon\, \sigma\, u - C \leq 0$, whose positive root gives the maximum feasible cardinality
\begin{equation}
\label{eq:kstar}
    k^{\star} = \left\lfloor \left(
        \frac{-z_\epsilon\, \sigma + \sqrt{z_\epsilon^2\, \sigma^2 + 4\mu\, C}}{2\mu}
    \right)^{\!2}\;\right\rfloor.
\end{equation}
The optimization problem therefore reduces to selecting the $k^{\star}$ items with the largest values $v_i$. The constraint determines \emph{how many} items to select; the choice of \emph{which} items is a combinatorial optimization within the set of $k^{\star}$-element subsets. Note, in general, except in this simplified case, the optimal cardinality $k^{\star}$ is not known beforehand.

The initial state $\lvert+\rangle^{\otimes n}$ (see \autoref{fig:ansatz}) distributes amplitude across Hamming-weight sectors according to the binomial distribution $\mathbb{P}(k) = \binom{n}{k}/2^n$, which is sharply peaked around $k = n/2$ with the bitstring subspace $\mathcal{H}_k = \operatorname{span}\bigl\{\lvert x\rangle : \lvert x\rvert = k\bigr\}$. When the feasible cardinality $k^{\star}$ is far from this peak, as is typical under restrictive capacity constraints, the amplitude initially assigned to $\mathcal{H}_{k^{\star}}$ is exponentially suppressed. Consequently, an effective mixer must facilitate the redistribution of amplitude between different Hamming-weight sectors in order to realize feasible solutions.

The necessary components for this redistribution are contained within the single-qubit and two-qubit interactions of $H_\mathrm{{mixer}}$. The single-qubit $X$ term allows transitions between $\mathcal{H}_k$ and $\mathcal{H}_{k \pm 1}$, while the $\XX$ term
enables swaps within a fixed Hamming-weight sector ($\lvert 01\rangle \leftrightarrow \lvert 10\rangle$, $\Delta k = 0$) and between sectors differing by two ($\lvert 00\rangle \leftrightarrow \lvert 11\rangle$, $\Delta k = \pm 2$). Together, these components promote amplitude propagation between Hamming-weight sectors ($\mathcal{H}_k$ to $\mathcal{H}_{k\pm 2}$), enabling access to the entire solution space.

In \autoref{app:extensionToFullXX}, we extend our toy-example analysis in \autoref{sec:toy-example} to $\XX$ interactions on a linear geometry. We deduce that for the state $\lvert \mathcal{T}\rangle$, which denotes the computational-basis state corresponding to the $k^{\star}$-element subset $\mathcal{T} \subseteq \{1,\dots,n\}$, the measurement probability becomes

\begin{align}
\label{eq:full-chain-prob_main}
    \mathbb{P}(\mathcal{T}) \;\approx\; |\omega_{\mathcal{T}}|^2\;\biggl(1 \;+\; 2\gamma\!\!\sum_{\substack{(j,\,j') \in E(G):\\j\in \mathcal{T},\; j'\notin \mathcal{T}}} \!\!\sin\!\Bigl(2\bigl(\phi_j - \phi_{j'}\bigr)\Bigr)\biggr),
\end{align}
where $\phi_j=\alpha v_j - \beta \mu_j$ denotes the effective rotation, $j \in \mathcal{T}$ is the included item, and $j' \notin \mathcal{T}$ is its excluded neighbor on the linear mixer graph.

In essence, equation~\eqref{eq:full-chain-prob_main} shows that the $\XX$ chain acts as a bank of parallel reduced-cost comparators.
Each edge $(j,\,j')$ in the mixer graph contributes a term $\sin\!\Bigl(2\bigl(\phi_j - \phi_{j'}\bigr)\Bigr)$ that tests whether keeping item~$j$ versus swapping it for item~$j'$ is favorable, given the current value--weight trade-off set by $\alpha$ and $\beta$.
If the included item has a more favorable value--weight balance ($\alpha v_j - \beta\mu_j > \alpha v_{j'} - \beta\mu_{j'}$), the contribution is positive and the probability of~$\mathcal{T}$ increases.
Conversely, when the excluded item is favored, amplitude flows away from~$\mathcal{T}$ toward the swapped configuration.
The total probability boost for a subset~$\mathcal{T}$ is the sum over all boundary edges, defined as edges with one endpoint in~$\mathcal{T}$ and one outside.
States whose included items consistently outperform their excluded neighbors accumulate positive contributions from multiple edges simultaneously, concentrating amplitude on high-quality solutions.

When a capacity constraint $\sum_{j\in \mathcal{T}}\mu_j \le C$ is present, the amplitude transfers on different boundary edges are no longer independent.
The marginal residual $\hat{r}(\mathcal{T}) = \sum_{j\in \mathcal{T}}\mu_j - C$ limits how much weight can be redistributed across the boundary of~$\mathcal{T}$: a swap that replaces item~$j$ with a heavier item~$j'$ is only feasible if $\mu_{j} - \mu_{j'} \ge \hat{r}(\mathcal{T})$.
At the circuit level, the phase separator encodes this through the weight penalty~$\beta$, where edges whose swap would violate the capacity bound receive a large negative phase shift, automatically suppressing their amplitude contribution in Eq.~\eqref{eq:full-chain-prob_main}.
Thus, the same mechanism identified in the two-state case, where the interplay between the phase gap~$\delta$ and the marginal residual $\hat{r}$, extends to the full chain. However, here the boundary edges collectively share a single capacity budget, coupling their amplitude transfers through the aggregate feasibility condition $\sum_{j\in \mathcal{T}}\mu_j \le C$.


\subsubsection{Circuit Parameters}
\paragraph{Loss function.} Although the structure of the ansatz naturally encodes feasibility through its phase and mixing dynamics, its performance ultimately depends on how the variational parameters are selected and optimized. Practically, key considerations include training stability and expressivity, and stability in particular is known to introduce barren plateau issues in variational schemes \cite{Larocca_2025}. Therefore, variational training is restricted to relatively small circuits with a loss function defined as

\begin{equation}
\label{eq:loss}
f(\bm x_{\bm \theta}) = -\bm v^\top\bm x_{\bm \theta} + \lambda \max(0, \sum_{m=1}^M L_m(\bm x_{\bm \theta}; \bm\mu_m, \bm\sigma_m) - C^{(m)}),
\end{equation}
where 
\begin{equation}
    L_m(\bm x_{\bm \theta}; \bm\mu_m, \bm\sigma_m) := \sum_{i=1}^n \mu_i^{(m)} x_i +
                    \Phi^{-1}(1-\epsilon^{(m)})\sqrt{\sum_{i=1}^n {\sigma_i^{(m)}}^2 x_i}
\end{equation} 
refers to the left hand side in \eqref{eq:socp-constraint}.  


Instead of minimizing the empirical mean, the optimizer employs conditional value‑at‑risk (CVaR$(\xi)$) \cite{barkoutsos2020improving} computed over samples collected from the quantum processor, with parameters updated using COBYLA \cite{Powell1994,2020SciPy-NMeth}. In general, CVaR may be defined with respect to either the lower or upper $\xi$‑tail of the cost distribution, corresponding to an emphasis on the most favorable or most adverse outcomes, respectively. CVaR$(\xi)$-based optimization is empirically observed to demonstrate improved convergence characteristics and enhanced robustness to hardware noise \cite{barkoutsos2020improving, Barron_2024}. In this work, we adopt the lower‑tail CVaR.

A close examination of the loss in \eqref{eq:loss} shows that the penalty term always contributes positively to the loss, allowing the multiplier $\lambda$ to be chosen as a sufficiently large positive value. However, the constraints may still be violated because the first term, $-\bm v^\top\bm x_{\bm \theta}$, is negative. Setting $\lambda = \sum_i v_i$ guards against the largest possible reduction in the loss that can be caused by the term $-\bm v^\top\bm x_{\bm \theta}$.

\paragraph{Parameter transfer.} For larger problem instances, optimized parameters are either directly transferred from smaller instances or employed to warm start a limited number of variational optimization steps on the larger circuits. The employed parameter transfer scheme takes into account characteristics of the constraints to align the transfer of parameters from circuits with fewer qubits ($\mathcal{C}_{\mathrm{src}}$) to circuits with a larger number of qubits ($\mathcal{C}_{\mathrm{tgt}}$). 

Each constraint $m$ is characterized by a feature vector 
\[
\bm{f}^{(m)} = \bigl[\bar{\mu}^{(m)},\; \bar{\sigma}^{(m)},\; r_C^{(m)}\bigr]^\top \in \mathbb{R}^3,
\]
constructed from per-constraint normalized weights. Specifically, the parameters $\mu_i^{(m)}$ and $\sigma_i^{(m)}$ are rescaled by the root-mean-square magnitude of $\{\mu_i^{(m)}\}_{i=1}^n$ within each constraint $m$, yielding scale-invariant statistics. Here, $\bar{\mu}^{(m)}$ and $\bar{\sigma}^{(m)}$ denote the mean normalized weight and standard deviation, and 
\[
r_C^{(m)} = \frac{\sum_{i=1}^n \mu_i^{(m)}}{C^{(m)}}
\]
is the capacity ratio, which remains invariant under this normalization.

Given feature matrices $\bm{F}_{\mathrm{src}}, \bm{F}_{\mathrm{tgt}} \in \mathbb{R}^{M \times 3}$, we solve the linear assignment problem $\pi^* = \arg\min_{\pi} \sum_{i=1}^{M} \|\bm{f}_{\mathrm{tgt}}^{(i)} - \bm{f}_{\mathrm{src}}^{(\pi(i))}\|_2$ via the Hungarian algorithm~\cite{Kuhn1955} in $\mathcal{O}(M^3)$ time, then permute the source circuit parameter $\bm{\beta}$ matrix by $\pi^*$ before transfer  (see \autoref{fig:ansatz}). This ensures each target constraint inherits parameters from the most structurally similar source constraint.

\subsection{Simulation and hardware runs}
With the circuit architecture and parameter transfer strategy established, the next step is to implement the resulting ansatz via matrix product state (MPS) simulations and gate‑based quantum hardware sampling. 

Simulations were carried out using the Qiskit \texttt{SamplerV2} \cite{qiskit-runtime} interface with the MPS backend \cite{vidal2003efficient, qiskit-aer}, which efficiently approximates circuit execution by truncating entanglement according to a fixed bond dimension. A maximum bond dimension of 512 and a truncation threshold of $10^{-6}$ were used to balance simulation accuracy and computational cost. All simulations matched the qubit count of the corresponding knapsack instances and employed fixed random seeds for reproducibility.

Hardware runs were conducted on IBM Heron processors, including \texttt{ibm\_fez}, \texttt{ibm\_pittsburgh}, and \texttt{ibm\_aachen}, with circuits characterized by 20-150 qubits that feature up to 3443 gates with total circuit depth up to 177. The two-qubit gates of the ansatz in \autoref{fig:ansatz} are arranged along a linear chain, and these entangling operations are then mapped to the native heavy-hex connectivity of the IBM quantum device. Additionally, the circuit is transpiled using the Qiskit Primitive V2 transpiler \cite{qiskit-runtime} with the optimization level set to 3. Dynamical decoupling is employed for error suppression using the default `\XX' sequence \cite{PhysRevApplied.20.064027}. The resulting transpiled circuit is then executed on the quantum processor using the Qiskit sampling primitive, with $N_{\mathrm{shots}}=2^{15}$ measurement samples collected.

\subsection{Self-consistent post-processing}
Samples produced by quantum hardware are inherently noisy, which can significantly degrade solution quality for constrained combinatorial optimization problems. To mitigate these effects, a novel problem‑aware classical post‑processing procedure is applied to the raw measurement outcomes. As described below and in pseudocode in \autoref{app:postProcessing}, this method iteratively refines noisy samples by exploiting problem structure and preferentially applying trial bit flips that lead to improved objective values while maintaining feasibility.

This method draws inspiration from the configuration recovery scheme used for error mitigation in quantum chemistry calculations \cite{Robledo-Moreno2025}. 
The procedure alternates between a residual-guided sample
repair step and updates to a \emph{selection-probability vector} $\bm{p}\in[0,1]^n$, where each component $p_i$ estimates the frequency with which item~$i$ appears in high-quality feasible solutions. This probability vector is
maintained by an outer iteration loop and fed back into the repair heuristic, closing a self-consistent feedback loop that progressively improves the recovery set. To make this procedure precise, the recovery problem can be formalized as follows.

\paragraph{Recovery problem.} Given the knapsack objective $f(\bm{x})=\bm{v}^\top\bm{x}$, following (\ref{eq:constraint-residual}) the constraint residual of constraint~$m$ is
\begin{equation}
r^{(m)}(\bm{x})
= 
\sum_{i=1}^{n}\mu_i^{(m)} x_i
+ \Phi^{-1}\!\bigl(1-\epsilon^{(m)}\bigr)\,
  \sqrt{\textstyle\sum_{i}(\sigma_i^{(m)})^2 x_i},
\;-\; C^{(m)}
\label{eq:residual}
\end{equation}
and a sample is feasible if and only if $r^{(m)}(\bm{x})\leq 0$ for all $m$.
Sample quality is measured by the penalty-augmented SOCP objective
\begin{equation}
    \mathcal{L}(\bm{x})=-f(\bm{x})+\lambda\sum_m g^{(m)}(\bm{x})
\end{equation}
where
$g^{(m)}(\bm{x}) = \max\bigl(0,r^{(m)}(\bm{x})\bigr)$ and
the repair heuristic can be viewed as greedily decreasing
$\mathcal{L}$ under single-bit moves.

\paragraph{Two-phase sample repair.}
At the start of each iteration, the sample pool is partitioned
into a feasible subset $\mathcal{S}_{\mathrm{feas}}$ and an
infeasible subset $\mathcal{S}_{\mathrm{infeas}}$ according to the
residuals~\eqref{eq:residual}. Every sample in
$\mathcal{S}_{\mathrm{feas}}\cup\mathcal{S}_{\mathrm{infeas}}$ is passed through a two-phase repair procedure that operates through single-bit flips. 

\emph{Phase~1 (feasibility restoration)} For samples drawn from $\mathcal{S}_{\mathrm{infeas}}$,
the first phase removes items by flipping selected bits from
$1$ to $0$ while any constraint is violated.
For each currently selected item~$i$, let $\bm{x}^{(-i)}$ denote the
vector obtained from $\bm{x}$ by setting $x_i$ to zero, and define
the relief as
\begin{equation}
\mathrm{relief}_i(\bm{x}) = \sum_m g^{(m)}(\bm{x}) -
\sum_m g^{(m)}(\bm{x}^{(-i)}).
\end{equation}
Among items with positive relief, Phase~1 flips the bit of the one
maximizing
\begin{equation}
\mathrm{score}_i
= \frac{\mathrm{relief}_i(\bm{x})}{\max(v_i,\,\delta)}
+ \eta\,(1-p_i),
\label{eq:repair-score}
\end{equation}
where $\delta>0$ is a small constant that prevents division by zero
for zero-value items and $\eta\geq 0$ sets the strength of the
probability bias.  The first term rewards relief per unit of lost
profit; the second nudges the heuristic toward flipping items that
are rarely selected in the current probability model (i.e., those
with low $p_i$).
If no bit flip yields positive relief, the phase terminates early,
leaving the sample unchanged.

\emph{Phase~2 (greedy improvement)} Phase~2 applies to every sample that
is feasible at the end of Phase~1:  samples that skipped Phase~1
drawn from $\mathcal{S}_{\mathrm{feas}}$ and
those repaired from $\mathcal{S}_{\mathrm{infeas}}$. For each sample, Phase~2 flips
unselected bits from $0$ to $1$ in decreasing order of $v_i$,
accepting each flip only if the resulting sample remains feasible. Since the square-root term in~\eqref{eq:residual} is non-separable, to determine feasibility, all $M$ constraints are checked. In other words, the marginal contribution of flipping item~$i$ from
$0$ to $1$ in constraint~$m$ is
\begin{equation}
\delta r_i^{(m)}(S) =
\mu_i^{(m)}+\Phi^{-1}(1-\epsilon^{(m)})[\sqrt{S^{(m)}+(\sigma_i^{(m)})^2}
-\sqrt{S^{(m)}}], \\
\end{equation}
with $S^{(m)}=\sum_j(\sigma_j^{(m)})^2 x_j$.

\paragraph{Objective-weighted probability update.}
After repair, feasible samples receive softmax weights
$w_k \propto \exp(f(\bm{x}_k)-f_{\max})$, where subtracting the
maximum profit $f_{\max}$ ensures numerical stability.  The
selection probabilities are then updated by an average
over $K$ random batches
$\mathcal{B}_1,\dots,\mathcal{B}_K$ of the samples:
\begin{equation}
p_i^{\mathrm{new}}
= \frac{1}{K}\sum_{\ell=1}^{K}
  \frac{\sum_{\bm{x}\in\mathcal{B}_\ell} w_{\bm{x}}\,x_i}
       {\sum_{\bm{x}\in\mathcal{B}_\ell} w_{\bm{x}}},
\qquad i=1,\dots,n.
\label{eq:prob-update}
\end{equation}
Batching reduces the variance of $\bm{p}^{\mathrm{new}}$ relative to
a single global weighted mean and makes the estimate more robust
when the feasible set is small.  If no repaired sample is feasible, as may occur during intial iterations for highly noisy pools, the full repaired set is used as a fallback to ensure that~\eqref{eq:prob-update} remains well-defined.

\paragraph{Self-consistent loop.}
The outer loop repeatedly (i)~repairs every sample in the current
pool, (ii)~recomputes weights and probabilities on the feasible
subset, and (iii)~re-partitions the repaired pool into feasible and
infeasible subsets for the next iteration.  No fresh samples are
drawn from $\bm{p}^{\mathrm{new}}$; the scheme is a deterministic
refinement of the original pool of samples, with the probabilities
entering the next repair step only through the
score~\eqref{eq:repair-score}.  Iteration stops when
$\Delta = \tfrac{1}{n}\|\bm{p}^{\mathrm{new}}-\bm{p}\|_1 <
\epsilon_{\mathrm{conv}}$ or after $\mathcal{T}$ iterations, with per-iteration cost $O(|\mathcal{S}|\cdot(F_{\mathrm{rem}}+F_{\mathrm{add}})\cdot
n\cdot M)$.  

\subsection{Computational cost.}

The quantum-classical scheme involves three steps. 

The \textbf{first} step, using COBYLA, optimizes circuit parameters $\bm \theta$ with an MPS simulator on a relatively small, surrogate problem with $N^\mathrm{S}_{\mathrm{eval}}$ function evaluations yielding $\bm \theta^S$ after convergence. 

The \textbf{second} step transfers $\bm \theta^S$ to the larger circuit corresponding to the target problem to be solved and subsequently tunes, again with COBYLA, the parameters with  $N_{\mathrm{eval}}$ ($N_{\mathrm{eval}} <\!\!< N^\mathrm{S}_{\mathrm{eval}}$) function evaluations either with an approximate MPS simulator or on quantum hardware yielding $\bm \theta^*$. 
This is followed by a final sampling run with the circuit initialized to $\bm \theta^*$.  
As a result, the large, target circuit's execution time is $\propto N_{\mathrm{shots}}N_{\mathrm{eval}}$ while the parameter training on the surrogate problem is a one-time $\propto N^{S}_{\mathrm{shots}}N^\mathrm{S}_{\mathrm{eval}}$, as the same parameters can be transferred to circuits corresponding to multiple target problems sharing similar features.

Finally, the \textbf{third} step consists of classical recovery run for $\mathcal{T}$ iterations using $|\mathcal{S}|$ unique bitstrings\footnote{There may be repetitions among the $N_{\mathrm{shots}}$ samples.} to refine noisy hardware samples obtained from the final sampling run on the circuit in the second step. 
As detailed in the prior section, the recovery step has a runtime  $ \propto \mathcal{T}\cdot|\mathcal{S}|$. 
Each recovery iteration performs $|\mathcal{S}|$ repair calls, each of which
evaluates at most $\mathcal{O}\bigl((F_{\mathrm{rem}}+F_{\mathrm{add}})\cdot
M\cdot n\bigr)$ constraint checks with a corresponding probability update at the cost of
$\mathcal{O}(|\mathcal{S}|\cdot n)$.


Overall, as hardware sample quality improves through additional parameter training and reduced noise, the recovery step is expected to become less resource‑intensive. Higher‑fidelity sample distributions reduce recovery costs in terms of both sample complexity and iteration count. The number of function evaluations $N_{\mathrm{eval}}$, which limits the extent of circuit refinement during recovery, therefore serves as a key control parameter for the total runtime.

\section{Results} \label{sec:results}

\begin{table}[h]
    \centering
    \begin{tabular}{|c|c|c|} \hline
        $n$ & $\mathbb{P}(t_o< 180)$ & $\bar{t}_r$ \\ \hline
         20 & 100\% & 0.1 \\
         40 & 97\% & 10 \\
         75 & 69\% & 73  \\
         100 & 56\% & 103 \\
         150 & 36\% & 125 \\ \hline
    \end{tabular}\hspace{0.1in}
    \begin{tabular}{|c|c|c|} \hline
        $M$ & $\mathbb{P}(t_o< 180)$ & $\bar{t}_r$ \\ \hline
         1 & 98\% & 7 \\
         5 & 87\% & 33 \\
         20 & 68\% & 71 \\
         50 & 55\% & 95 \\
         100 & 50\% & 106 \\ \hline
    \end{tabular}\\[0.1in]
    \begin{tabular}{|c|c|c|} \hline
        $\mu_{\max}$ & $\mathbb{P}(t_o< 180)$ & $\bar{t}_r$ \\ \hline
         5 & 52\% & 94 \\
         $\{25,50,100\}$ & 76\% & 51 \\ \hline
    \end{tabular}
    \caption{Statistics for average classical benchmarking runtimes $\bar{t}_r$ and the time taken to prove optimality $t_o$ using Gurobi, both times measured in seconds. $\mathbb{P}(t_o< 180)$ denotes the proportion of instances that were solved to proven optimality within 180 seconds. The optimality gap is defined as $f_{\mathrm{ub}} / f_{\mathrm{best}} - 1$, where $f_{\mathrm{ub}}$ and $f_{\mathrm{best}}$ denote the best upper bound and the best known feasible solution, respectively.}
    \label{tab:stats}
\end{table}

\begin{figure*}
 \begin{subfigure}[t]{.45\linewidth}
    \includegraphics[width=\linewidth]{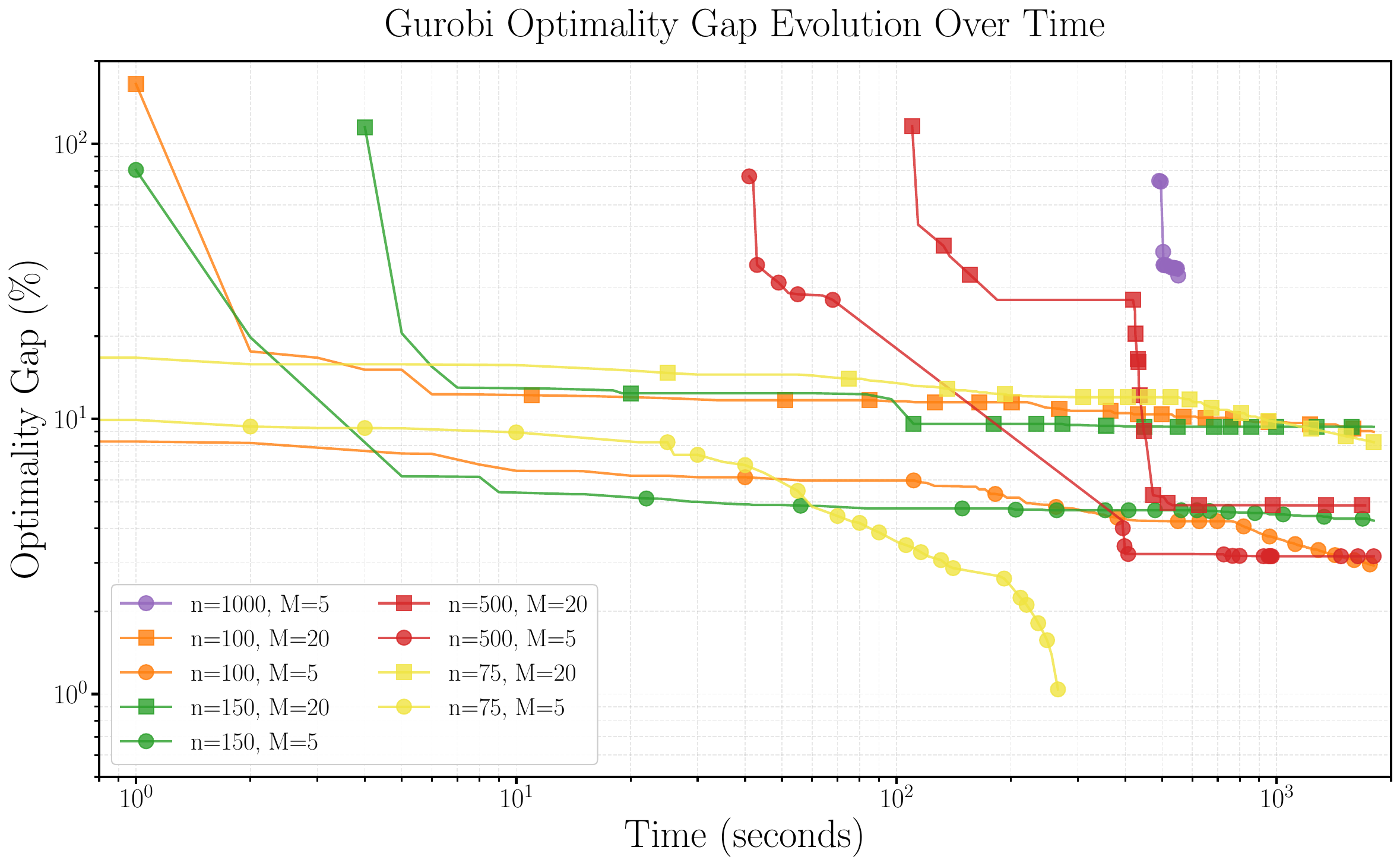}%
  \caption{}
  \label{fig: gurobi}
  \end{subfigure}\hfill                
\begin{subfigure}[t]{0.50\linewidth}
  \includegraphics[width=\linewidth]{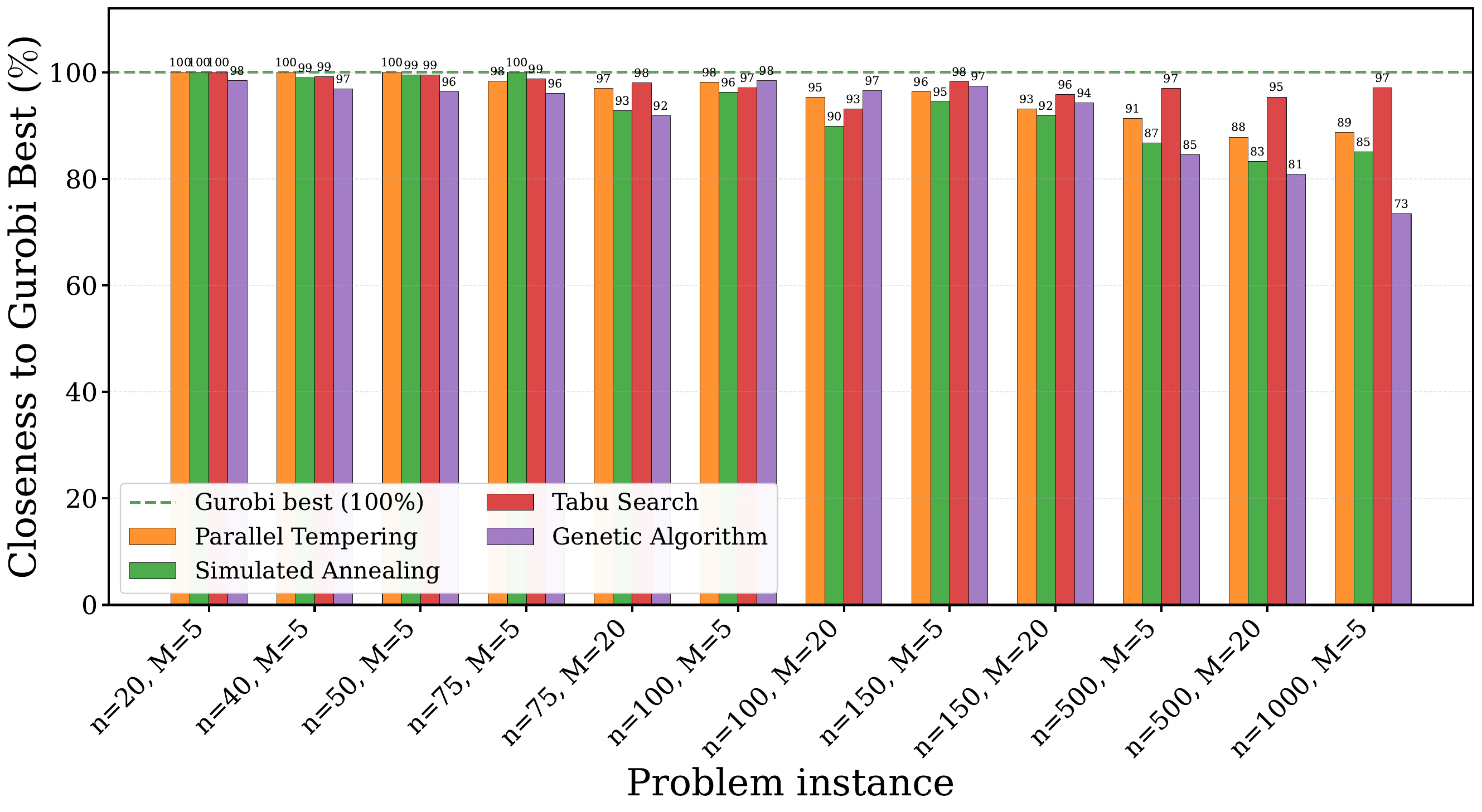}
  \caption{}
  \label{fig: heuristic}
\end{subfigure}
\label{fig:classical solver}
\caption{(a) Gurobi optimality gap evolution over time for different problem instances when Gurobi runtime is capped at 30 minutes. The plot shows gap convergence on a log-log scale, where different colors represent problem sizes ($n$) and marker shapes indicate the number of constraints ($M$): circles for $M=5$, squares for $M=20$. (b) Final solution quality comparison across all solvers after 1800 seconds for different problem instances. The y-axis shows closeness ($\kappa$) to Gurobi's best solution as a percentage, where 100\% (green dashed line) indicates matching Gurobi's performance. Note, to obtain the ``Gurobi-best" solutions, the 30 minute time limit was relaxed. Results are sorted by $n$ and $M$.}
\label{fig: classical benchmarking}
\end{figure*}

\paragraph{Classical benchmarking.} \autoref{tab:stats} presents average Gurobi runtime trends observed over different problem sizes $n$, number of constraints $M$, and maximum item weight mean $\mu_{\mathrm{max}}$ for the 3,900 instances discussed in section \autoref{sec:instances}. Overall, the results show a pronounced runtime inflection point at approximately 75 items to reach optimality, and problems with more constraints exhibit longer runtimes compared to those with fewer constraints.  

\autoref{fig: gurobi} presents Gurobi optimality gap evolution over time, capped to a maximum runtime of 30 minutes. These results are reported for ``hard" instances obtained by following a systematic scan over different parameters as described in \autoref{sec:instances}. The figure demonstrates that relatively large problems (i.e, $n\in\{500,1000\}$) maintain higher gaps even after 30 minutes, indicating difficulty continuously increasing with problem size. Conversely, for small knapsack instances, Gurobi consistently proves optimality within 30 minutes. For more difficult instances, the Gurobi-best solution is typically identified in under 30 minutes; however, extended runtimes and manual intervention--primarily through parameter tuning--are often required to reduce the optimality gap. For the most challenging instances, substantial user intervention is necessary to obtain the Gurobi-best solution, and runtimes up to 30 hours may be required to further tighten the gap. Although runs could be extended beyond this limit, diminishing returns were observed when attempting to fully close the optimality gap.

While Gurobi serves as an exact baseline, classical hardness is further characterized by comparing its performance against that of several classical heuristic solvers optimizing the same loss function (see Eq.~\eqref{eq:loss}). \autoref{fig: heuristic} indicates closeness $\kappa$ of heuristic solutions to the best solutions produced by Gurobi, where $\kappa = f_{\mathrm{solver}}/f_{\mathrm{best}} \times 100$ is defined as the best-so-far item-value ratio of the best feasible bitstring obtained by the classical heuristics $f_{\mathrm{solver}}$ to the best $f_{\mathrm{best}}$ obtained using Gurobi. The heuristics recover the proven optimum on the smallest instances ($n \leq 50$, $M=5$) and the best-performing heuristic remains within 5\% of the Gurobi reference on every instance considered. Performance nevertheless degrades with increasing item count and, more markedly, with the number of constraints: at $n \geq 500$ the spread across solvers widens considerably (e.g., 73--97\% at $n=1000$, $M=5$), with tabu search proving the most robust ($\kappa \geq 93\%$ throughout). Notably, across all problem instances, heuristics never exceed the Gurobi-best objective.

{
\setlength{\tabcolsep}{4.5pt}
\begin{table}[!hb]
\centering
\caption{Performance and circuit resources of the quantum scheme on knapsack instances of size $n$, with 5 constraints in most cases, except the instance labeled $100(20)$, which includes 20 constraints. Closeness $\kappa$ (\%) is measured against Gurobi's best solution; the quantum value is from the best sampled bitstring and the classical value is from the best-performing heuristic among Parallel Tempering, Tabu Search, Genetic Algorithm, and Simulated Annealing. Circuit metrics are the number of variational layers $p$, total circuit depth, two-qubit gate depth, and total gate count. For $n < 100$, $p$ was chosen by sweeping layer counts and selecting the value that minimized CVaR; for $n \geq 100$, $p$ was fixed by transferring parameters optimized on the 50-item instance.}
\label{tab:combined}
\begin{tabular}{c cc c cccc}
\toprule
& \multicolumn{2}{c}{Closeness $\kappa$ (\%)} & & \multicolumn{4}{c}{Circuit resources} \\
\cmidrule{2-3} \cmidrule{5-8}
$n$ & Quantum & \makecell{Classical \\ heuristic} & & $p$ & \makecell{Total\\ depth} & \makecell{2Q \\depth} & \makecell{Gate\\ count} \\
\midrule
20  & $100.0$ & $100.00$ & & 9 & 55  & 35  & 571  \\
40  & $100.0$ & $100.00$ & & 8 & 71  & 53  & 1032 \\
50  & $100.0$ & $100.00$ & & 7 & 77  & 61  & 1143 \\
75  & $98.6$  & $100.00$           & & 10 & 114  & 92  & 2390 \\
\midrule
100 & $97.2$           & $98.5$  & & 7 & 127 & 111 & 2293 \\
100(20) & $96.9$           & 96.5  & & 7 & 127 & 111 & 2293 \\
150 & $94.9$  & $98.3$           & & 7 & 177 & 161 & 3443 \\
\bottomrule
\end{tabular}
\footnotetext{Best classical heuristic per row: Parallel Tempering for $n \in \{20, 40, 50, 100\}$; Tabu Search for $n=75$; Simulated Annealing for $n=150$.}
\end{table}
}

\paragraph{Quantum-classical scheme.} \autoref{tab:combined} summarizes the performance of the proposed quantum scheme across knapsack instances of varying size, alongside the circuit resources required for each configuration. Solution quality is reported in terms of closeness $\kappa$ to Gurobi's optimum, enabling a direct comparison between the best sampled bitstring from the quantum solver and the strongest classical heuristic identified per instance among Parallel Tempering, Tabu Search, Genetic Algorithm, and Simulated Annealing.


For small instance sizes ($n \leq 50$), both the quantum scheme and the classical heuristic recover Gurobi’s best ($\kappa = 100\%$), indicating that these regimes are easily tractable for both approaches. At larger scales, the classical heuristic generally maintains a slight edge in solution quality. At $n = 75$, the classical method still achieves $\kappa = 100\%$, while the quantum scheme attains $98.6\%$. A similar trend persists at $n=100$ with 5 constraints and $n=150$, where the classical heuristic reaches $98.5\%$ and $98.3\%$, compared to $97.2\%$ and $94.9\%$ for the quantum scheme, respectively. However, the quantum approach marginally outperforms the classical heuristic in the more tightly constrained $n=100$ case with 20 constraints ($96.9\%$ vs.\ $96.5\%$). These results indicate that the quantum scheme remains competitive as problem size grows suggests that the quantum scheme is particularly competitive as constraint density increases. 

The circuit depths reported in \autoref{tab:combined} induce a rapidly increasing classical simulation cost when using MPS methods, since the runtime scales polynomially with both the circuit size and the MPS bond dimension \cite{Vidal_2003, schieffer2025harnessingcudaqsmpstensor}. For circuits of fixed depth, the dominant contributions to the runtime is governed by the bond dimension alone. Here, the circuit employed contains a linear, one-dimensional entanglement structure whose bond dimension grows linearly with the number of items in the knapsack instance. As a result, problem instances exceeding approximately 150 items quickly become prohibitive for classical MPS-based simulation, which contrasts the relatively constant execution times of the same circuits on quantum hardware. Consequently, while larger instances beyond those considered in this work become classically intractable, the corresponding quantum circuits remain practically executable, provided sufficiently capable hardware is available.


\begin{figure}[!htb]
    \centering
    \includegraphics[trim={0cm 0cm 0cm 1.5cm}, clip, width=0.8\linewidth]{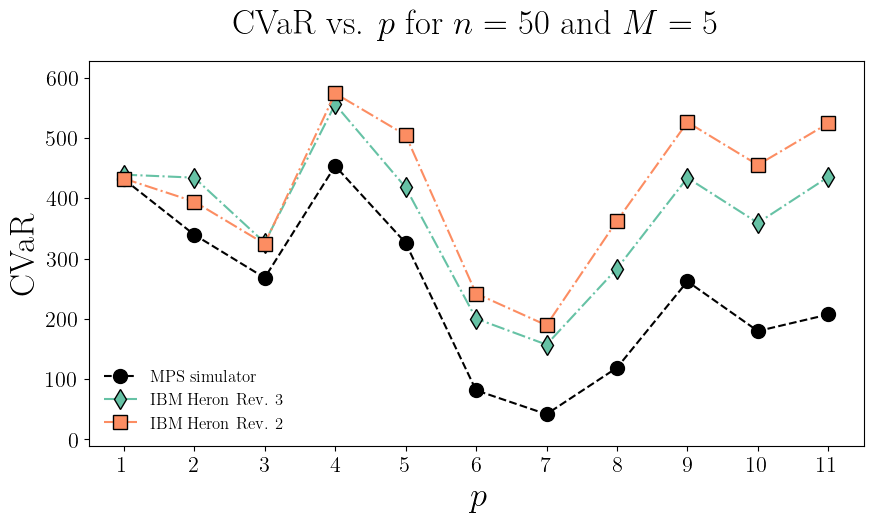}
    \caption{CVaR as a function of circuit depth $p$ for a 50 item, 5 constraint knapsack instance.}
    \label{fig:cvar_hardware}
\end{figure}

To investigate hardware constraints, \autoref{fig:cvar_hardware} compares the CVaR as a function of circuit depth $p$ across the MPS simulator and two IBM Heron processors, Revision 2 and Revision 3, for a 50 item problem with 5 constraints. Overall, increasing $p$ up to $p=7$ improves CVaR, though the trend is non-monotonic, consistent with the preceeding analysis. Across all $p$, hardware implementations exhibit higher CVaR values than their noiseless simulation counterparts, as expected due to noise. Beyond $p >7$, hardware noise begins to dominate, leading to unreliable CVaR values and, in some cases, degradation in the performance of the scheme.

\begin{figure}
    \centering
    \includegraphics[width=0.7\linewidth]{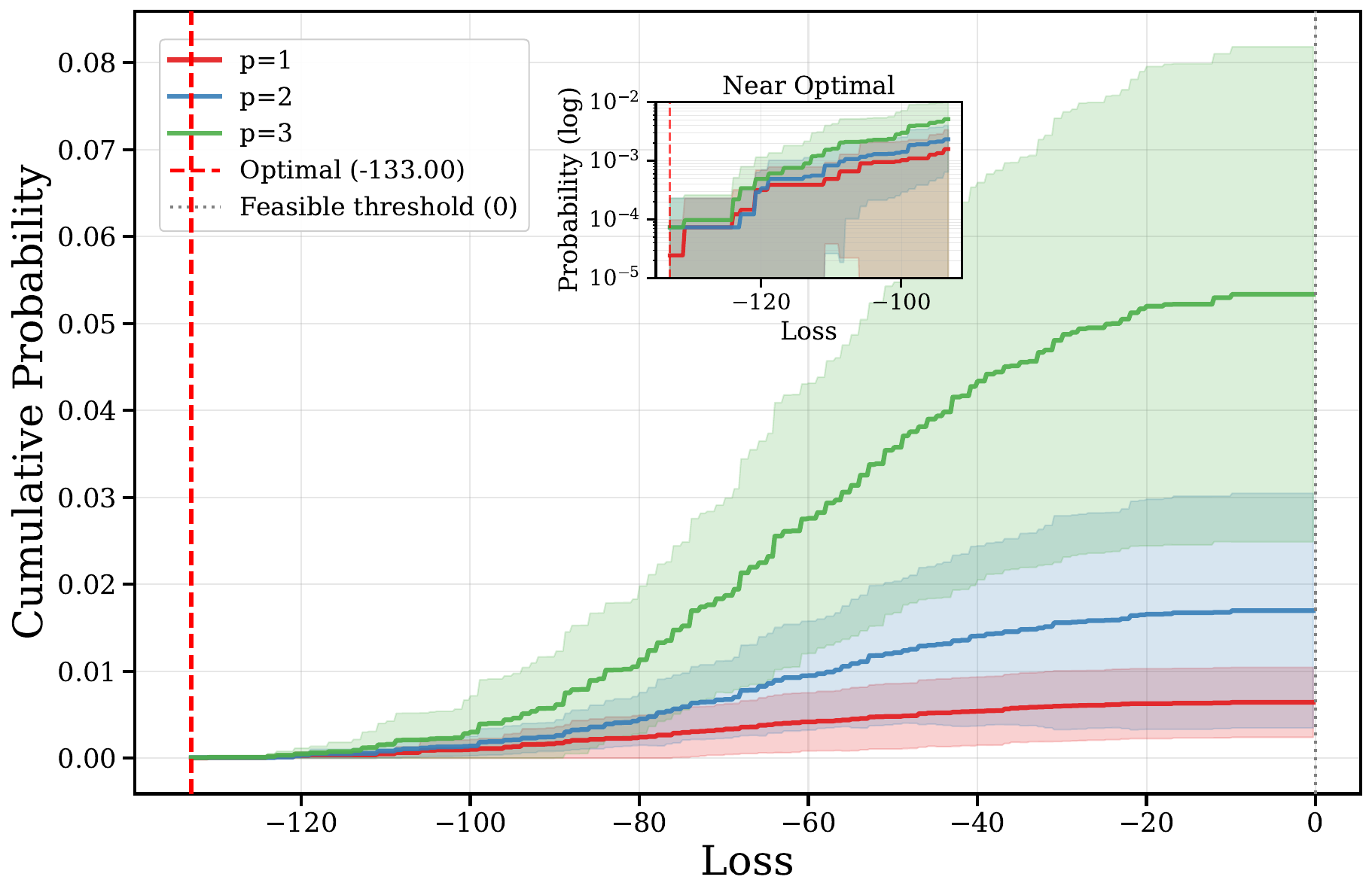}
    \caption{Empirical CDF of the loss $L$ for circuit depths $p = 1, 2, 3$ on the 20 item, 5 constraint CCKP instance. The plot displays the mean CDF with shaded
confidence regions ($\pm 1$ standard deviation over 10 independent runs of the variational scheme) for the feasible region ($L < 0$) under the chosen penalty $\lambda$. An inset with logarithmic y-axis focuses on the near-optimal region
($-140 \leq L \leq -93$), highlighting the probability of achieving solutions close to the optimum. The vertical dashed line at $L = -133$ marks the optimal solution, while the dotted line at $L = 0$ indicates the feasibility threshold.}
    \label{fig:p_dependence}
\end{figure}

\paragraph{Empirical cumulative distribution function.} \autoref{fig:p_dependence} depicts the empirical CDF constructed from the sampled solutions. Feasible samples are identified by negative loss values under Eq.~\eqref{eq:loss}, whereas samples with positive loss are infeasible. The samples were collected from a variational quantum circuit trained using MPS simulator corresponding to a 20 item problem with 5 constraints. The standard deviation was computed by averaging 10 independent runs of the
variational scheme. As the circuit depth and corresponding circuit complexity increases, the likelihood of obtaining feasible solutions ($L < 0$) improves:
for $p=1$ the probability is $0.6\% \pm 0.4\%$, for $p=2$ it is $1.7\% \pm 1.4\%$, and for $p=3$ it reaches $5.4\% \pm 2.9\%$. Additionally, the inset in \autoref{fig:p_dependence} shows that the optimal solution at loss $L = -133$ is sampled with measurement probability of $\sim 10^{-4}$ for $p=3$.


\begin{figure*}
 \begin{subfigure}[t]{.50\linewidth}
    \includegraphics[trim={0cm 0cm 0cm 1.45cm}, clip, width=\linewidth]{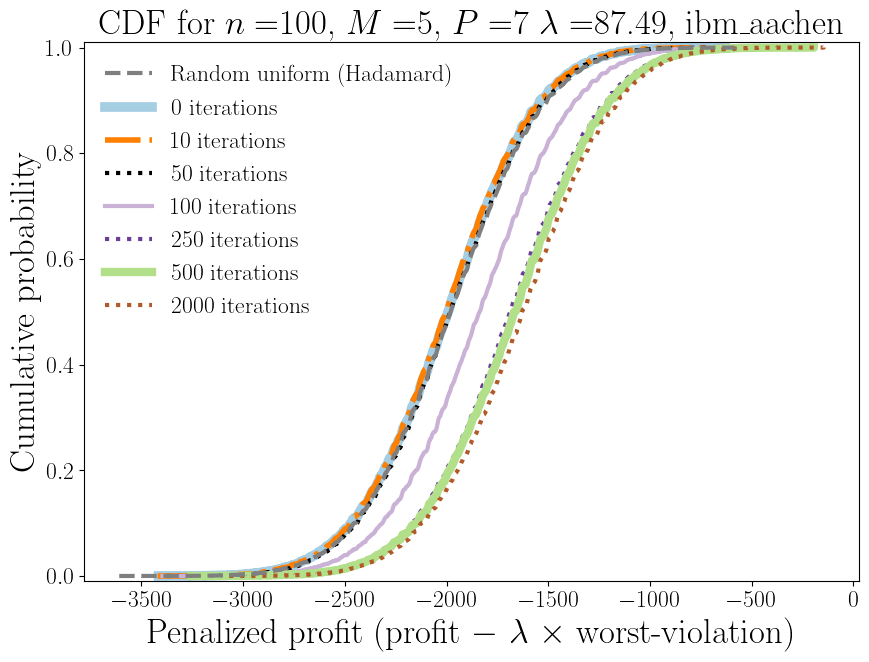}%
  \caption{}
  \label{fig: cdf evolution iter}
  \end{subfigure}\hfill                
\begin{subfigure}[t]{0.50\linewidth}
  \includegraphics[trim={0cm 0cm 0cm 0cm}, clip, width=\linewidth]{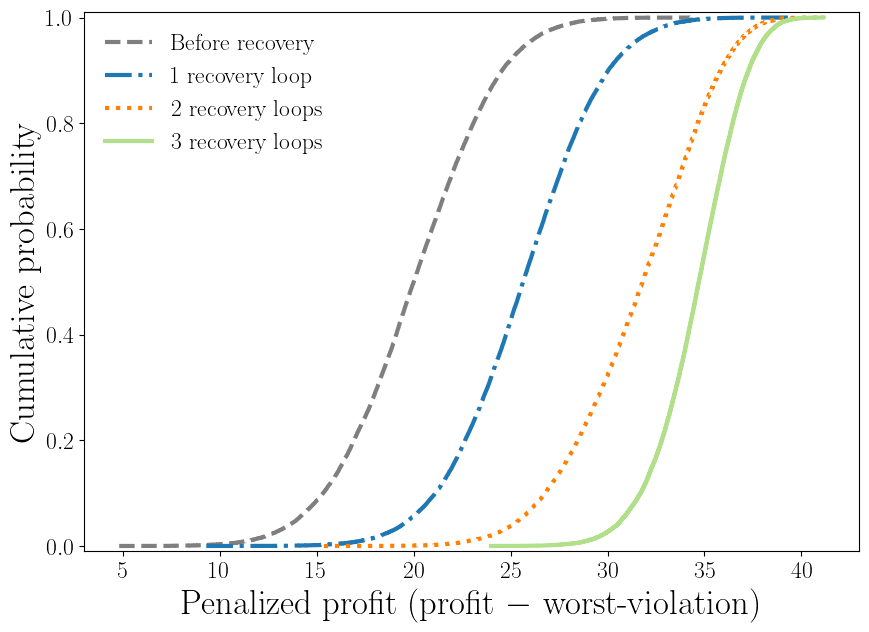}
  \caption{}
  \label{fig: cdf recovered evolution}
\end{subfigure}
\caption{Empirical cumulative distributions of penalized profit for a 100‑item, 5‑constraint knapsack problem obtained from hardware‑sampled quantum circuits. (a) Comparison of a random (Hadamard) baseline and increasing numbers of fine‑tuning iterations following parameter transfer from a 50‑item, 5‑constraint knapsack. Profit is penalized using a $\lambda$‑scaled worst‑case SOCP constraint violation, where $\lambda$ matches the value employed during CVaR‑based training. Successive fine‑tuning stages induce a progressive rightward shift of the distributions that saturates at 500 iterations (green), indicating increasing concentration of probability mass on higher‑quality solutions under the training objective. (b) Empirical distributions of feasibility‑aware profit before and after applying a fixed number of recovery‑algorithm iterations. Here, profit is penalized by the unscaled worst‑case SOCP violation, reflecting direct constraint enforcement independent of the CVaR training penalty.}
\label{fig: cdf evolution}
\end{figure*}

\begin{figure}[!htbp]          
    \centering
    \begin{subfigure}[b]{0.48\textwidth}
        \centering
        \includegraphics[trim={18cm 0cm 4cm 10cm}, clip, width=1.05\textwidth]{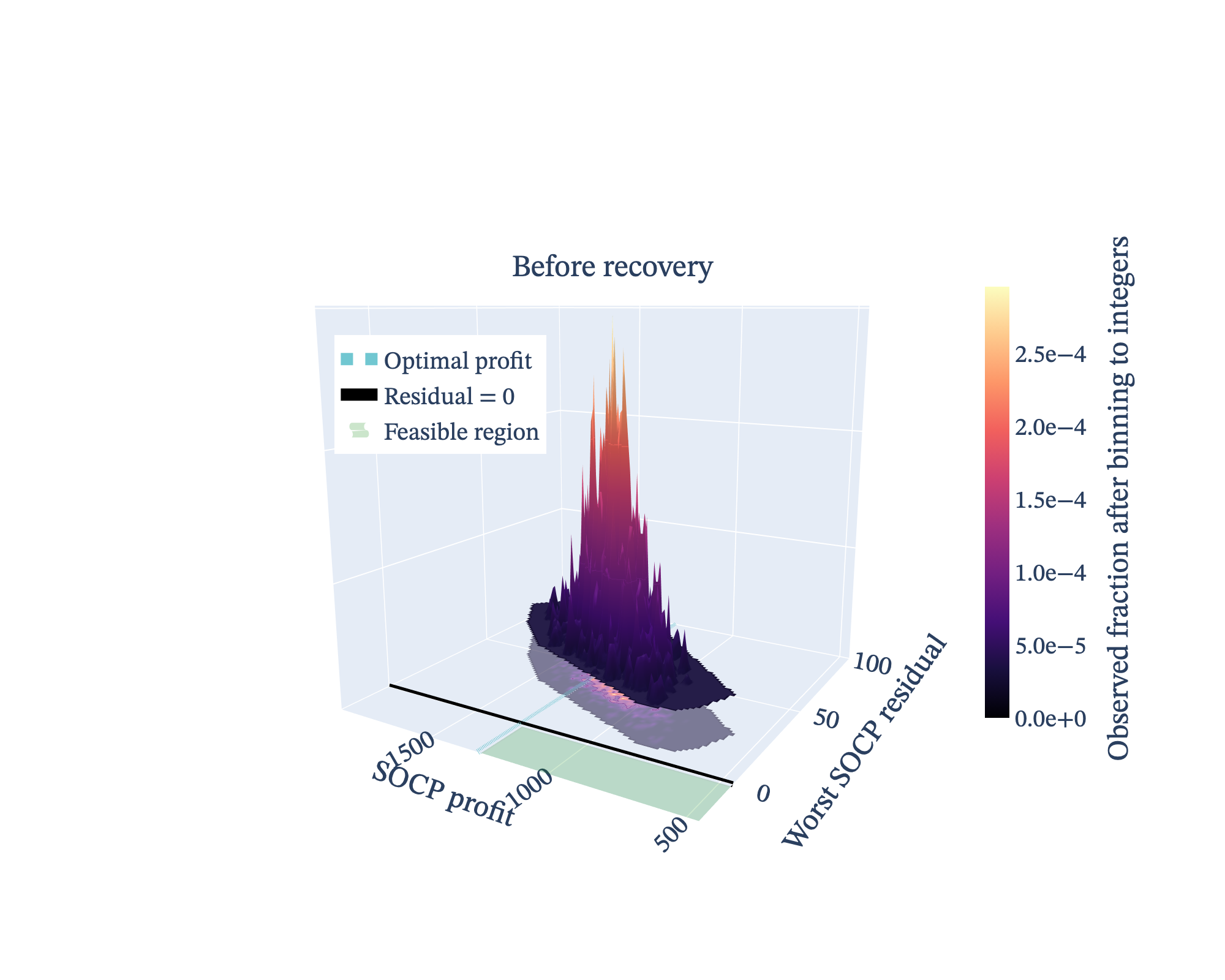}   
        \vspace{0.2cm}                                   
        \includegraphics[trim={0cm 0cm 0cm 3cm}, clip, width=\textwidth]{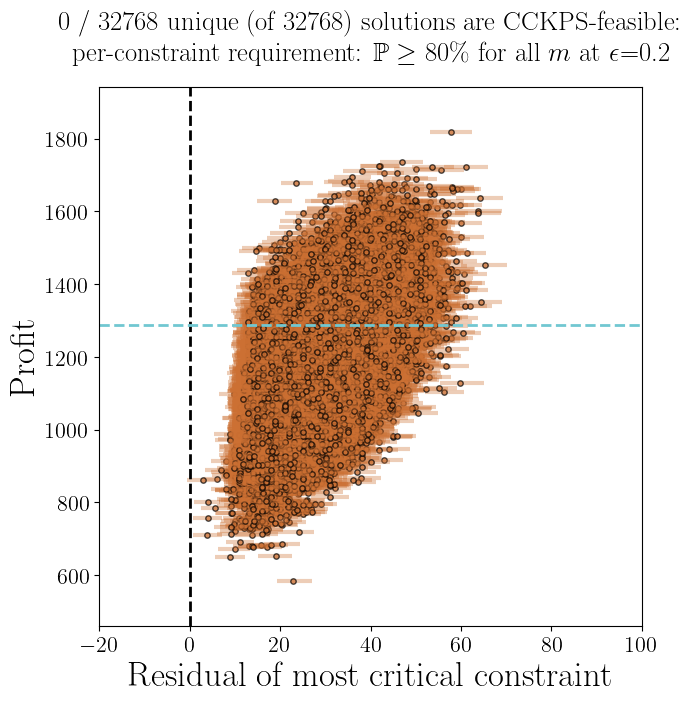}   
        \caption{}
        \label{fig:panelA}
    \end{subfigure}
    \hfill   
    \begin{subfigure}[b]{0.48\textwidth}
        \centering
        \includegraphics[trim={18cm 0cm 4cm 0cm}, clip, width=1.05\textwidth]{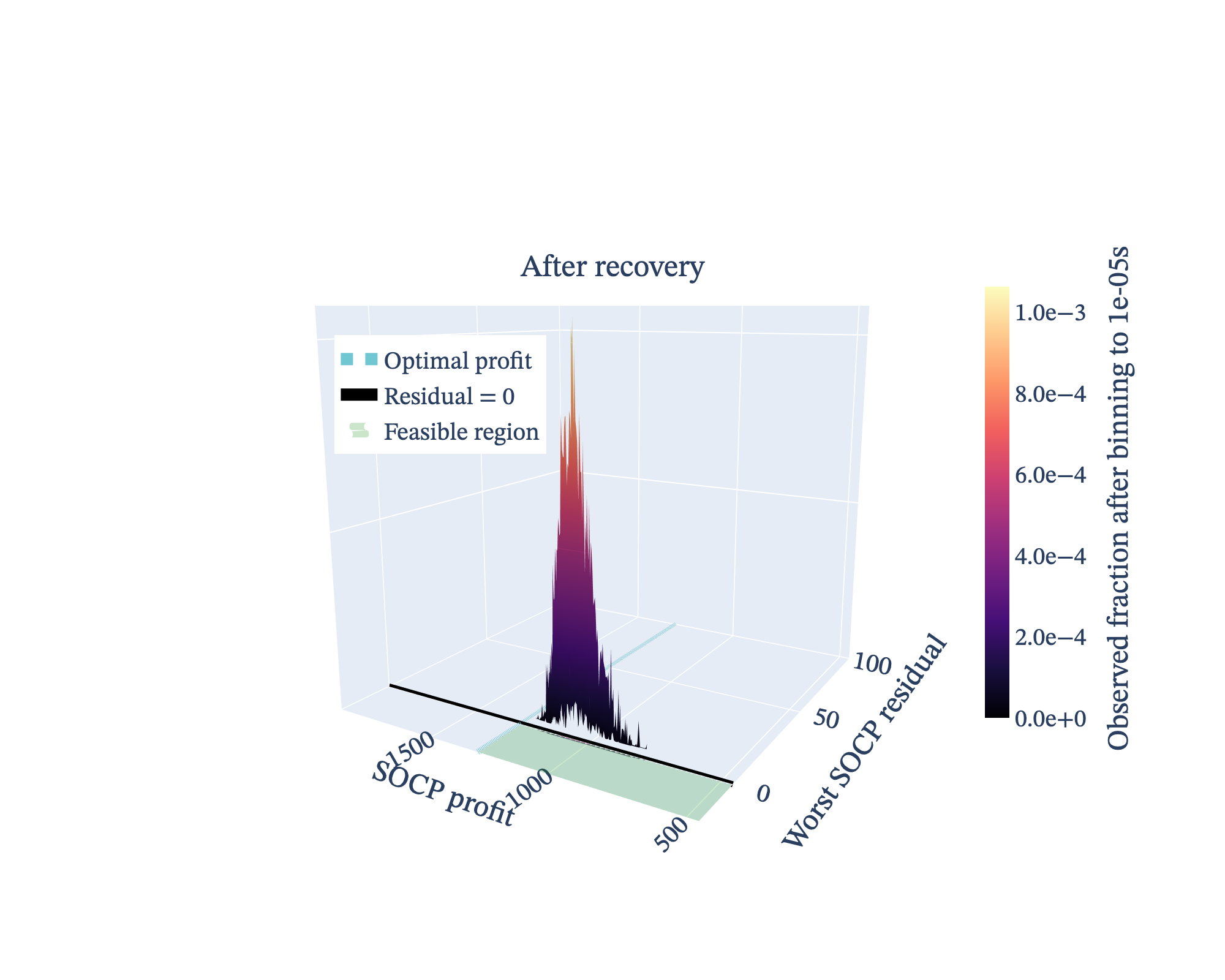}   
        \vspace{0.2cm}
        \includegraphics[trim={0cm 0cm 0cm 3cm}, clip, width=\textwidth]{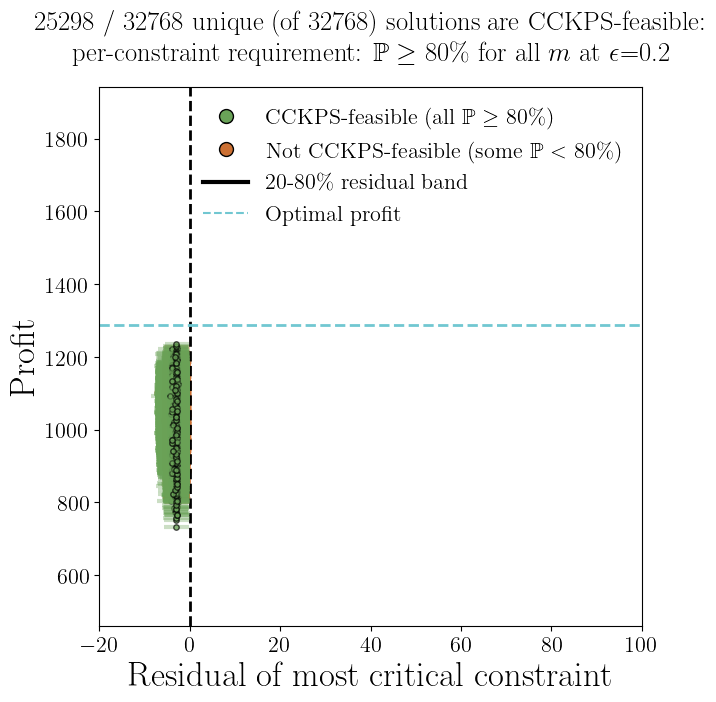}   
        \caption{}
        \label{fig:panelB}
    \end{subfigure}

\caption{Visualizations of observed bitstrings for a 100‑item, 5‑constraint knapsack instance (a) before and (b) after recovery. Top: SOCP surface plots of bitstrings showing profit versus worst SOCP residual, with the colorbar representing observed hardware sampling probability. The blue line represents the classical best solution profit, and its intersection with the black line ($r^{(*)}=0$) outlines the SOCP-feasible zone in the $xy$-plane (green rectangle). Bottom: CCKP scatter plots of bitstrings showing profit versus CCKP critical residuals with the highest violation probability. The dashed blue line once again indicates the classically best profit. Each point corresponds to a unique bitstring colored by CCKP-sampled (CCKPS) feasibility, and error bars denote the 20–80\% range ($\epsilon = 0.2$) of sampled residuals $r^{(*)}$ from 1000 band samples. After recovery, nearly all bitstrings are CCKPS-feasible, approaching equivalence of CCKPS-feasibility to SOCP-feasibility under sufficient band samples.}
    \label{fig:residual}
\end{figure}

Furthermore, \autoref{fig: cdf evolution} demonstrates how the loss function evolves across different optimization steps for circuits with 100 qubits, corresponding to a problem instance with 100 items and 5 constraints. \autoref{fig: cdf evolution iter} illustrates the evolution of the loss function during the circuit fine-tuning step, which is displayed as empirical CDFs of penalized profit. Samples were obtained from hardware after classical fine-tuning of varying maximum number of iterations\footnote{Using COBYLA, the term iteration here implies a function evaluation.} with the MPS simulator. Specifically, the case of 0 iterations represents a circuit initialized by direct parameter transfer from a 50‑qubit instance, whereas iteration counts greater than 0 reflect additional fine‑tuning steps carried out using the MPS simulator. The calculation of the penalized profit was set by subtracting a $\lambda$‑weighted worst‑case SOCP constraint violation, with $\lambda$ fixed to the value used during CVaR-based training to ensure consistency between the training objective and post‑hoc evaluation. 



The warm‑start (0 iteration) empirical distribution initially performs similarly to a random uniform baseline, with its effectiveness depending sensitively on the choice of circuit depth $p$ used for parameter transfer. Parameter transfer nonetheless provides a clear benefit by stabilizing and accelerating CVaR convergence, suggesting improved guidance of the optimization trajectory even when immediate gains in solution quality are limited. As shown in \autoref{fig: cdf evolution iter}, the fine-tuned distributions begin to outperform the random baseline after approximately 10 iterations, as indicated by a rightward shift in the CDF. However, for fewer than 100 fine-tuning iterations, these improvements remain relatively modest
, partially because COBYLA expects to perform at least $P\cdot(2+M)+2$ function evaluations.

As additional fine‑tuning stages are applied, the distributions exhibit a monotonic rightward shift, indicating a systematic reallocation of probability mass toward more profitable solutions. This improvement begins to plateau around 500 iterations, beyond which no further changes in the CDF are observed. The observed behavior reflects improved constraint satisfaction in the tail of the sampled distribution, together with increased alignment with the risk‑averse training objective imposed by the CVaR loss. 

\autoref{fig: cdf recovered evolution} illustrates the evolution of feasibility-aware profit distributions under successive applications of the recovery algorithm to the 100 item, 5 constraint knapsack after fine-tuning iterations plateaued at 500 iterations. In this figure, penalized profit aligns with the recovery penalty, which, in contrast to \autoref{fig: cdf evolution iter}, does not incur a scaling coefficient to the worst-case SOCP violation. The initial distribution corresponds to solutions directly obtained from hardware sampling and contains nonzero mass violating constraints. With increasing recovery iterations, the violating mass progressively redistributes toward higher penalized profit values, indicating the iterative correction of infeasible bitstrings. The dominant transformation occurs after the first recovery loop, which induces the largest rightward shift in the distribution. Subsequent recovery, particularly at three loops, contributes to variance reduction, tightening the distribution while shifting its mean to a lesser extent. After three loops, all bitstrings are SOCP-feasible, terminating the recovery step.

\paragraph{SOCP and CCKP relationship.} To compare feasibility assessed under the SOCP relaxation and under CCKP sampling, \autoref{fig:residual} displays visualizations of hardware-sampled results before and after the application of sample recovery on the 100 item, 5 constraint knapsack. The two surface plots depict SOCP-feasibility for bitstrings before and after applying the recovery algorithm and demonstrate how the initially broad distributions produced by hardware sampling systematically coalesces probability mass near the intersection of low constraint residuals and the classically determined optimal profit. Complementary, the two scatter plots assess feasibility under CCKP sampling (CCKPS), with sample weights drawn from the 20–80\% range of the corresponding CCKP distribution with 1000 band samples. The band sampling range was selected as $[\epsilon, 1-\epsilon]$ to match the SOCP formulation. 

Taken together, these surface and scatter plots enable a direct comparison between SOCP‑feasible and CCKPS‑feasible solutions. After recovery, the majority of bitstrings agree on feasibility under both formulations, and as the number of CCKPS band samples increases, this agreement improves monotonically. This convergence indicates that dense CCKP sampling recovers the same feasible set identified by the SOCP relaxation, while the recovery procedure drives the distribution toward constraint‑saturating, high‑profit configurations near the SOCP/CCKPS-feasibility boundary.

\section{Conclusions and Outlook}
\label{sec:conclusions}
In this work, we introduced a novel quantum–classical framework for solving CCKPs, motivated by stochastic decision-making challenges that arise in insurance underwriting and related domains. 
By reformulating probabilistic constraints into deterministic equivalents and leveraging a QAOA-based variational quantum circuit, we demonstrated how current quantum devices can be used as effective heuristic samplers for stochastic optimization. 
The proposed scheme integrates variational training, parameter transfer across problem sizes, and problem-aware classical post-processing to overcome known challenges such as barren plateaus \cite{Larocca_2025, mcclean2018barren} and hardware noise. Overall, empirical results show that the quantum-generated samples concentrate probability mass on feasible, high-quality solutions, while maintaining moderate circuit depths.

These findings indicate a promising regime in which our proposed quantum-classical scheme can potentially yield substantial gains over purely classical approaches. Although the method consistently produces higher‑quality solutions, the fine‑tuning stage currently requires a relatively high number of optimizer iterations. Beyond parameter transfer, alternative ``offline" circuit‑parameter training strategies may help reduce this cost and will be explored in future work. Moreover, the ability to address larger, more constrained problems is expected to widen the performance gap between classical and quantum solvers, motivating continued investigation of scalable solution methods.

Beyond the specific knapsack formulation studied here, the methods point toward a more general paradigm for applying quantum computing to real-world stochastic optimization problems. 
Combining a specific structure in quantum circuits, risk-aware objective formulations, and robust classical post-processing provides a flexible template that can be adapted to other chance-constrained and uncertainty-driven decision problems. 
Although current hardware limitations restrict problem sizes and achievable depths, the results presented here suggest a pathway for scaling as quantum devices and error-mitigation techniques improve. 
From an industry perspective, this work highlights how hybrid quantum–classical approaches can be meaningfully aligned with business-relevant risk–reward trade-offs, offering a promising direction for future research and experimentation at the intersection of quantum computing and enterprise optimization.
\bibliographystyle{unsrt}
\bibliography{ref}
\newpage
\appendix
\section*{Appendix}
\label{appendix}
\addcontentsline{toc}{section}{Appendix}
\section{Derivation} \label{app:derivation}
We restate Theorem~\ref{theorem:1} for convenience.

\knapsacktheorem*

\begin{proof}
Let the total weight of the selected items corresponding to a constraint $m$ be $B^{(m)} = \sum_{i=1}^{n} \tilde{w}_i^{(m)} x_i$. Each term $\tilde{w}_i^{(m)} x_i$ is either $0$ or $\tilde{w}_i^{(m)}$, and the random variables $\tilde{w}_i^{(m)}$ are independent. Since the distribution $\Delta$ is closed under convolution, the random variable $B^{(m)}$ also follows the distribution $\Delta$.
Because $x_i \in \{0,1\}$, it follows that $x_i^2 = x_i$, and thus the variance of $B^{(m)}$ simplifies to $\operatorname{Var}(B^{(m)}) = \sum_{i=1}^{n} {\sigma_i^{(m)}}^2 x_i$. Consequently,

$$B^{(m)} \sim \Delta\bigl(\mu^{(m)}(\bm{x}), {\sigma^{(m)}}^2(\bm{x})\bigr)$$
where
$\mu^{(m)}(\bm{x}) = \sum_{i=1}^n \mu^{(m)}_i x_i$, 
and
$\sigma^{(m)}(\bm{x}) = \sqrt{\sum_{i=1}^n {\sigma_i^{(m)}}^2 x_i}$.

We impose the chance constraint
\[
\mathbb{P}(B^{(m)} \leq C^{(m)}) \geq 1-\epsilon.
\]

By standardizing $B^{(m)}$, we obtain
\[
\frac{B^{(m)} - \mu^{(m)}(\bm{x})}{{\sigma^{(m)}}(\bm{x})} \sim \Delta(0,1).
\]
Hence
\begin{align*}
\mathbb{P}(B^{(m)} \leq C^{(m)}) 
&= 
\mathbb{P}\!\left( \frac{B^{(m)} - \mu^{(m)}(\bm{x})}{{\sigma^{(m)}}(\bm{x})} 
\leq 
\frac{C^{(m)} - \mu^{(m)}(\bm{x})}{{\sigma^{(m)}}(\bm{x})} \right) \\
&= 
\Phi\!\left( \frac{C^{(m)} - \mu^{(m)}(\bm{x})}{{\sigma^{(m)}}(\bm{x})} \right),
\end{align*}
where $\Phi(\cdot)$ denotes the cumulative distribution function (CDF) of the standard normal distribution. Additionally, we define

$$b := \frac{C^{(m)} - \mu^{(m)}(\bm{x})}{\sigma^{(m)}(\bm{x})}.$$ 
Then, the constraint $\Phi(b) \geq 1-\epsilon$ implies
\[
b \geq \Phi^{-1}(1-\epsilon)
\quad \Rightarrow \quad
\frac{C^{(m)} - \mu^{(m)}(\bm{x})}{{\sigma^{(m)}}(\bm{x})} \geq \Phi^{-1}(1-\epsilon),
\]
where this inversion is valid because the CDF $\Phi$ is strictly increasing, which allows us to apply its inverse to both sides of the inequality $\Phi(b) \geq 1-\epsilon$.

Assuming $\sigma(\bm{x}) > 0$ and $\Phi^{-1}(1-\epsilon) > 0$ (equivalently, $1-\epsilon > 0.5$), we can rearrange the inequality to obtain
\[
C^{(m)} - \mu^{(m)}(\bm{x}) \geq \Phi^{-1}(1-\epsilon)\,{\sigma^{(m)}}(\bm{x}),
\]
that is,
\[
\sum_{i=1}^n \mu^{(m)}_i x_i 
+ 
\Phi^{-1}(1-\epsilon) \sqrt{\sum_{i=1}^n {\sigma_i^{(m)}}^2 x_i} 
\leq C^{(m)}.
\]

\end{proof}

\section{Extension to the full \texorpdfstring{$\XX$}{XX} chain}
\label{app:extensionToFullXX}

The two-state analysis \autoref{sec:toy-example} extends to the full mixer acting
on all $\binom{n}{k^{\star}}$ states in the feasible sector.
We work within the Hamming-weight sector $\mathcal{H}_{k^{\star}}$ and derive the
\emph{conditional} probability of measuring a given $k^{\star}$-subset, given
that the measurement outcome has Hamming weight~$k^{\star}$.
In our scheme, the penalty term in the loss function concentrates amplitude in
the feasible sector during optimization; the analysis below characterizes how
the mixer redistributes that amplitude among feasible solutions.

\smallskip\noindent\textit{State after the phase separator.}\;
Let $\mathcal{T} \subseteq \{1,\dots,n\}$ denote a $k^{\star}$-element subset, and let
$\lvert \mathcal{T}\rangle$ denote the corresponding computational-basis state.
Since the phase separator is diagonal with effective rotation
$\phi_j = \alpha v_j - \beta \mu_j$ on qubit~$j$, the eigenvalue associated with
$\mathcal{T}$ is
\begin{equation}
\label{eq:hp-eigenvalue}
    E_\mathcal{T} \;=\; \sum_{j \notin \mathcal{T}} \phi_j \;-\; \sum_{j \in \mathcal{T}} \phi_j .
\end{equation}
Starting from a uniform superposition over $k^{\star}$-subsets,
\[
\lvert\psi_0\rangle = \sum_\mathcal{T} \omega_{\mathcal{T}} \lvert \mathcal{T}\rangle,
\qquad
\omega_\mathcal{T} = \binom{n}{k^{\star}}^{-1/2},
\]
the phase separator produces
\begin{equation}
    \lvert\psi_1\rangle = \sum_\mathcal{T} \omega_\mathcal{T}\, e^{-i E_\mathcal{T}} \lvert \mathcal{T}\rangle .
\end{equation}
Probabilities are unchanged:
$\lvert\langle \mathcal{T} \lvert\psi_1\rangle\rvert^2 = |\omega_\mathcal{T}|^2$
for all~$\mathcal{T}$.

\smallskip\noindent\textit{First-order mixer expansion.}\;
The projected mixer within $\mathcal{H}_{k^{\star}}$ is
$H_{\mathrm{mix}}^{(k^{\star})} = -\sum_{(j,j') \in E(G)} T_{jj'}$,
where $T_{jj'}$ is the swap operator and $E(G)$ is the edge set of the mixer graph.
Expanding the mixer unitary to first order in~$\gamma$,
\begin{equation}
\label{eq:first-order-mixer}
    \lvert\psi_2\rangle
    = e^{-i\gamma H_{\mathrm{mix}}^{(k^{\star})}}\lvert\psi_1\rangle
    \;\approx\;
    \lvert\psi_1\rangle
    + i\gamma \sum_{(j,j') \in E(G)} T_{jj'}\lvert\psi_1\rangle .
\end{equation}

\smallskip\noindent\textit{Amplitude of a target subset~$\mathcal{T}$.}\;
We compute $\langle \mathcal{T} \lvert\psi_2\rangle$ by evaluating
$\langle \mathcal{T} \lvert T_{jj'} \lvert\psi_1\rangle$ for each edge.
For a fixed target $\mathcal{T}$ and edge $(j,j')$, the only source subset $\mathcal{T}'$ such that
$T_{jj'}\lvert \mathcal{T}'\rangle = \lvert \mathcal{T}\rangle$
is $\mathcal{T}' = T_{jj'}(\mathcal{T})$.
There are two cases:
\begin{enumerate}
    \item \textbf{Non-boundary edge} (both $j,j' \in \mathcal{T}$ or both
    $j,j' \notin \mathcal{T}$): the swap acts as the identity on
    $\lvert \mathcal{T}\rangle$, so $\mathcal{T}'=\mathcal{T}$ and
    \begin{equation}
        \langle \mathcal{T} \lvert T_{jj'} \lvert\psi_1\rangle
        = \omega_\mathcal{T}\, e^{-i E_\mathcal{T}}.
    \end{equation}

    \item \textbf{Boundary edge} ($j \in \mathcal{T},\; j' \notin \mathcal{T}$): the swap maps
    $\lvert \mathcal{T}'\rangle \to \lvert \mathcal{T}\rangle$ with
    \[
    \mathcal{T}' = \mathcal{T} \setminus \{j\} \cup \{j'\},
    \]
    giving
    \begin{equation}
        \langle \mathcal{T} \lvert T_{jj'} \lvert\psi_1\rangle
        = \omega_{\mathcal{T}'}\, e^{-i E_{\mathcal{T}'}} .
    \end{equation}
\end{enumerate}
For a boundary edge, the eigenvalue difference follows from
\eqref{eq:hp-eigenvalue}. In $\mathcal{T}$, item $j$ contributes $-\phi_j$ and $j'$ contributes
$+\phi_{j'}$, while in $\mathcal{T}'$ these roles are reversed. Therefore,
\begin{equation}
\label{eq:eigenvalue-diff}
    E_{\mathcal{T}'} - E_\mathcal{T} = 2(\phi_j - \phi_{j'}) .
\end{equation}

Collecting all edges and using $\omega_\mathcal{T}= \omega_{\mathcal{T}'}$,
\begin{equation}
\label{eq:collected-amplitude}
    \langle \mathcal{T} \lvert\psi_2\rangle
    = \omega_\mathcal{T}\, e^{-i E_\mathcal{T}}\biggl[
        1 + i\gamma\, n_0(\mathcal{T})
        + i\gamma \!\!\sum_{\substack{(j,j') \in E(G)\\ j\in \mathcal{T},\; j'\notin \mathcal{T}}}
        e^{-2i(\phi_j - \phi_{j'})}
    \biggr],
\end{equation}
where $n_0(\mathcal{T})$ is the number of non-boundary edges.

\smallskip\noindent\textit{Measurement probability.}\;
Taking the modulus squared to first order in~$\gamma$,
\begin{equation}
    \mathbb{P}(\mathcal{T})
    = \bigl|\langle \mathcal{T} \lvert\psi_2\rangle\bigr|^2
    \approx
    \omega_\mathcal{T}^2\Bigl(
        1 + 2\,\mathrm{Re}\bigl[
        i\gamma\, n_0(\mathcal{T})
        + i\gamma \textstyle\sum e^{-2i(\phi_j - \phi_{j'})}
        \bigr]
    \Bigr).
\end{equation}
The non-boundary term is purely imaginary and does not contribute.
For each boundary edge, using
$\mathrm{Re}(i e^{-i\theta}) = \sin\theta$,
\begin{equation}
    \mathrm{Re}\!\left(i\gamma\, e^{-2i(\phi_j - \phi_{j'})}\right)
    = \gamma\,\sin\!\bigl(2(\phi_j - \phi_{j'})\bigr).
\end{equation}
Substituting $\phi_j = \alpha v_j - \beta \mu_j$ yields
\begin{equation}
\label{eq:full-chain-prob}
    \mathbb{P}(\mathcal{T})
    \;\approx\;
    |\omega_\mathcal{T}|^2
    \biggl(
        1
        + 2\gamma
        \!\!\sum_{\substack{(j,\,j') \in E(G):\\ j\in \mathcal{T},\; j'\notin \mathcal{T}}}
        \sin\!\Bigl(
        2\bigl[\alpha(v_j - v_{j'})
        - \beta(\mu_j - \mu_{j'})\bigr]
        \Bigr)
    \biggr),
\end{equation}
where $j \in \mathcal{T}$ denotes an included item and $j' \notin \mathcal{T}$ its excluded neighbor on
the mixer graph.

\section{Post-Processing Algorithm} \label{app:postProcessing}
 
The post-processing scheme takes a set of binary
samples as input. These samples can be potentially noisy outputs from a quantum device, and the recovery scheme iteratively refines them through three stages: (i)~objective-based
weighting, (ii)~batched probability update, and (iii)~two-phase
sample repair.  The overall procedure is self-consistent: repaired
samples feed back into the probability estimates, which in turn guide
the next round of repairs.
 
\subsection{Objective-Based Weighting}
 
Feasible samples are assigned exponential (softmax) weights proportional
to their profit, so that better-quality solutions have more weight during the probability update step.
 
\begin{algorithm}[H]
\caption{Compute Objective-Based Weights}
\label{alg:weights}
\begin{algorithmic}[1]
\Require Feasible samples $\mathcal{S}_{\mathrm{feas}}$, profit
         function $f(\bm{x})=\bm{v}^\top\bm{x}$
\Ensure  Weight vector $\bm{w}$
\State $f_{\max} \gets \max_{\bm{x}\in\mathcal{S}_{\mathrm{feas}}} f(\bm{x})$
\For{each $\bm{x}_k \in \mathcal{S}_{\mathrm{feas}}$}
    \State $w_k \gets
      \dfrac{\exp\bigl(f(\bm{x}_k) - f_{\max}\bigr)}
            {\sum_{k'}\exp\bigl(f(\bm{x}_{k'}) - f_{\max}\bigr)}$
\EndFor
\State\Return $\bm{w}$
\end{algorithmic}
\end{algorithm}
 
\noindent
Subtracting $f_{\max}$ in the exponent ensures numerical stability, ensuring all exponents are negative.

\subsection{Batched Probability Update}
 
Bit-wise selection probabilities are updated via a batched weighted average.  Partitioning into $K$~batches and averaging the per-batch estimates provides variance reduction compared
to a single global average.
 
\begin{algorithm}[H]
\caption{Update Selection Probabilities}
\label{alg:probs}
\begin{algorithmic}[1]
\Require $\mathcal{S}_{\mathrm{feas}}$, weights $\bm{w}$,
         number of batches~$K$, number of items~$n$
\Ensure  Updated probabilities $\bm{p}^{\mathrm{new}}\in[0,1]^n$
\State Partition $\mathcal{S}_{\mathrm{feas}}$ into $K$ batches
       $\mathcal{B}_1,\dots,\mathcal{B}_K$
\For{$i = 1$ to $n$}
    \State $p_i^{\mathrm{new}} \gets
      \dfrac{1}{K}\sum_{\ell=1}^{K}
      \dfrac{\sum_{\bm{x}\in\mathcal{B}_\ell} w_{\bm{x}}\, x_i}
            {\sum_{\bm{x}\in\mathcal{B}_\ell} w_{\bm{x}}}$
\EndFor
\State\Return $\bm{p}^{\mathrm{new}}$
\end{algorithmic}
\end{algorithm}
 
\subsection{Two-Phase Sample Repair}
 Each sample undergoes a two-phase greedy repair procedure to increase feasibility and objective value among samples.
 
\begin{algorithm}[H]
\footnotesize
\caption{Two-Phase Sample Repair}
\label{alg:repair}
\begin{algorithmic}[1]
\Require Sample $\bm{x}$, probabilities $\bm{p}$, values $\bm{v}$,
         constraint data $(W,\bm{\sigma},\bm{C},\Phi^{-1})$,
         flip budgets $F_{\mathrm{rem}},F_{\mathrm{add}}$
\Ensure  Repaired sample $\bm{x}'$
\State $\bm{x}' \gets \bm{x}$
\State
\State \textbf{--- Phase 1: Feasibility Repair ---}
\For{$t = 1$ to $F_{\mathrm{rem}}$}
    \If{$\bm{x}'$ is feasible}
        \State \textbf{break}
    \EndIf
    \State $V_{\mathrm{tot}} \gets \sum_{m=1}^{M} g^{(m)}(\bm{x}')$
        \hfill\Comment{total violation}
    \State $i^* \gets -1$;\quad $s^* \gets -\infty$
    \For{each $i$ with $x'_i = 1$}
        \State $\bm{x}^{\mathrm{test}} \gets \bm{x}'$ with $x^{\mathrm{test}}_i \gets 0$
        \State $\mathrm{relief}_i \gets V_{\mathrm{tot}} -
               \sum_{m} g^{(m)}(\bm{x}^{\mathrm{test}})$
        \If{$\mathrm{relief}_i \leq 0$}
            \State \textbf{continue}
        \EndIf
        \State $\mathrm{score}_i \gets
               \dfrac{\mathrm{relief}_i}{\max(v_i,\,\delta)}
               + \eta\,(1 - p_i)$
            \hfill\Comment{prefer low-prob, low-value items}
        \If{$\mathrm{score}_i > s^*$}
            \State $i^* \gets i$;\quad $s^* \gets \mathrm{score}_i$
        \EndIf
    \EndFor
    \If{$i^* \geq 0$}
        \State $x'_{i^*} \gets 0$
    \Else
        \State \textbf{break}
            \hfill\Comment{no improving removal; stop}
    \EndIf
\EndFor
\State
\State \textbf{--- Phase 2: Greedy Improvement ---}
\If{$\bm{x}'$ is feasible}
    \State Sort unselected items by $v_i$ descending
    \State $t \gets 0$
    \For{each unselected item $i$ in sorted order}
        \If{$t \geq F_{\mathrm{add}}$}
            \State \textbf{break}
        \EndIf
        \State $\bm{x}^{\mathrm{test}} \gets \bm{x}'$ with $x^{\mathrm{test}}_i \gets 1$
        \If{$\bm{x}^{\mathrm{test}}$ is feasible}
            \State $x'_i \gets 1$;\quad $t \gets t + 1$
        \EndIf
    \EndFor
\EndIf
\State\Return $\bm{x}'$
\end{algorithmic}
\end{algorithm}
 
Phase~1 focuses exclusively on eliminating constraint violations by
removing items.  The scoring function
\[
\mathrm{score}_i
\;=\;
\frac{\mathrm{relief}_i}{\max(v_i,\,\delta)}
\;+\;
\eta\,(1-p_i)
\]
balances two considerations: (a)~maximizing constraint relief per unit
of lost profit, and (b)~preferring to remove items with low selection
probability~$p_i$, reflecting the collective signal from previously
observed high-quality samples.  Here, $\delta>0$ is a small constant to
avoid division by zero and $\eta\geq 0$ controls the strength of the
probability bias.
 
Phase~2 then exploits any remaining capacity slack by greedily adding
the highest-value items that fit within the constraints.  The SOCP constraint
\eqref{eq:socp-constraint} is non-separable since the variance term couples
all selected items, meaning that each candidate addition must be checked against the
full constraint set.
 
In the case where no single removal reduces the total violation, Phase~1 cannot fully restore feasibility. Therefore, the sample is returned as-is rather than replaced by
a random draw from~$\bm{p}$.  This preserves structural information
from the original quantum sample; infeasible samples are subsequently
filtered out during the probability update step.
\subsection{Main Self-Consistent Loop}
 
The three components are combined in an iterative loop, shown below, that
alternates between repairing samples and updating the probability
model until convergence.
 
\begin{algorithm}[H]
\caption{Self-Consistent Probability-Based Recovery}
\label{alg:main}
\begin{algorithmic}[1]
\Require Initial samples $\mathcal{S}_0$, profit function~$f$,
         batches~$K$, tolerance~$\epsilon_{\mathrm{conv}}$,
         max iterations~$\mathcal{T}$
\Ensure  Best feasible solution $\bm{x}^*$
\State Partition $\mathcal{S}_0 \to
       (\mathcal{S}_{\mathrm{feas}},\;\mathcal{S}_{\mathrm{infeas}})$
       by constraint satisfaction
\State Initialize $\bm{p}\in[0,1]^n$ from $\mathcal{S}_{\mathrm{feas}}$
       \hfill
\State $t \gets 0$
\Repeat
    \State \textbf{Repair:}
           Apply Algorithm~\ref{alg:repair} to every sample in
           $\mathcal{S}_{\mathrm{feas}}\cup\mathcal{S}_{\mathrm{infeas}}$
    \State Collect all repaired samples into $\mathcal{S}_{\mathrm{rec}}$
    \State Filter: $\mathcal{S}_{\mathrm{feas}}^{\mathrm{rec}} \gets
           \{\bm{x}\in\mathcal{S}_{\mathrm{rec}} :
             r^{(m)}(\bm{x})\leq 0\;\forall\,m\}$
    \If{$\mathcal{S}_{\mathrm{feas}}^{\mathrm{rec}} = \emptyset$}
        \State $\mathcal{S}_{\mathrm{feas}}^{\mathrm{rec}} \gets \mathcal{S}_{\mathrm{rec}}$
            \hfill\Comment{use all if none feasible}
    \EndIf
    \State $\bm{w} \gets
           \textsc{ComputeWeights}(\mathcal{S}_{\mathrm{feas}}^{\mathrm{rec}},\,f)$
           \hfill\Comment{Algorithm~\ref{alg:weights}}
    \State $\bm{p}^{\mathrm{new}} \gets
           \textsc{UpdateProbabilities}(\mathcal{S}_{\mathrm{feas}}^{\mathrm{rec}},\,
           \bm{w},\,K,\,n)$
           \hfill\Comment{Algorithm~\ref{alg:probs}}
    \State $\Delta \gets
           \frac{1}{n}\sum_{i=1}^{n}|p_i^{\mathrm{new}} - p_i|$
    \State Re-partition $\mathcal{S}_{\mathrm{rec}}$ into
           $\mathcal{S}_{\mathrm{feas}}$ and
           $\mathcal{S}_{\mathrm{infeas}}$
           by constraint satisfaction
    \State $\bm{p} \gets \bm{p}^{\mathrm{new}}$;\quad
           $t \gets t+1$
\Until{$\Delta < \epsilon_{\mathrm{conv}} \;\lor\; t \geq \mathcal{T}$}
\State\Return $\bm{x}^* \gets
        \arg\max_{\bm{x}\in\mathcal{S}_{\mathrm{feas}}} f(\bm{x})$
\end{algorithmic}
\end{algorithm}
 
\paragraph{Convergence.}
The loop terminates when the mean absolute change in selection
probabilities falls below a threshold~$\epsilon_{\mathrm{conv}}$, or
after a maximum of~$\mathcal{T}$ iterations. The repair step drives
infeasible samples toward feasibility, the weighting step concentrates
probability mass on high-profit regions, and the two mechanisms
reinforce each other.
 
\paragraph{Computational cost.}

\begin{figure*}
    \centering
    \includegraphics[trim={0cm 0cm 0cm 2cm}, clip,width=1.15\linewidth]{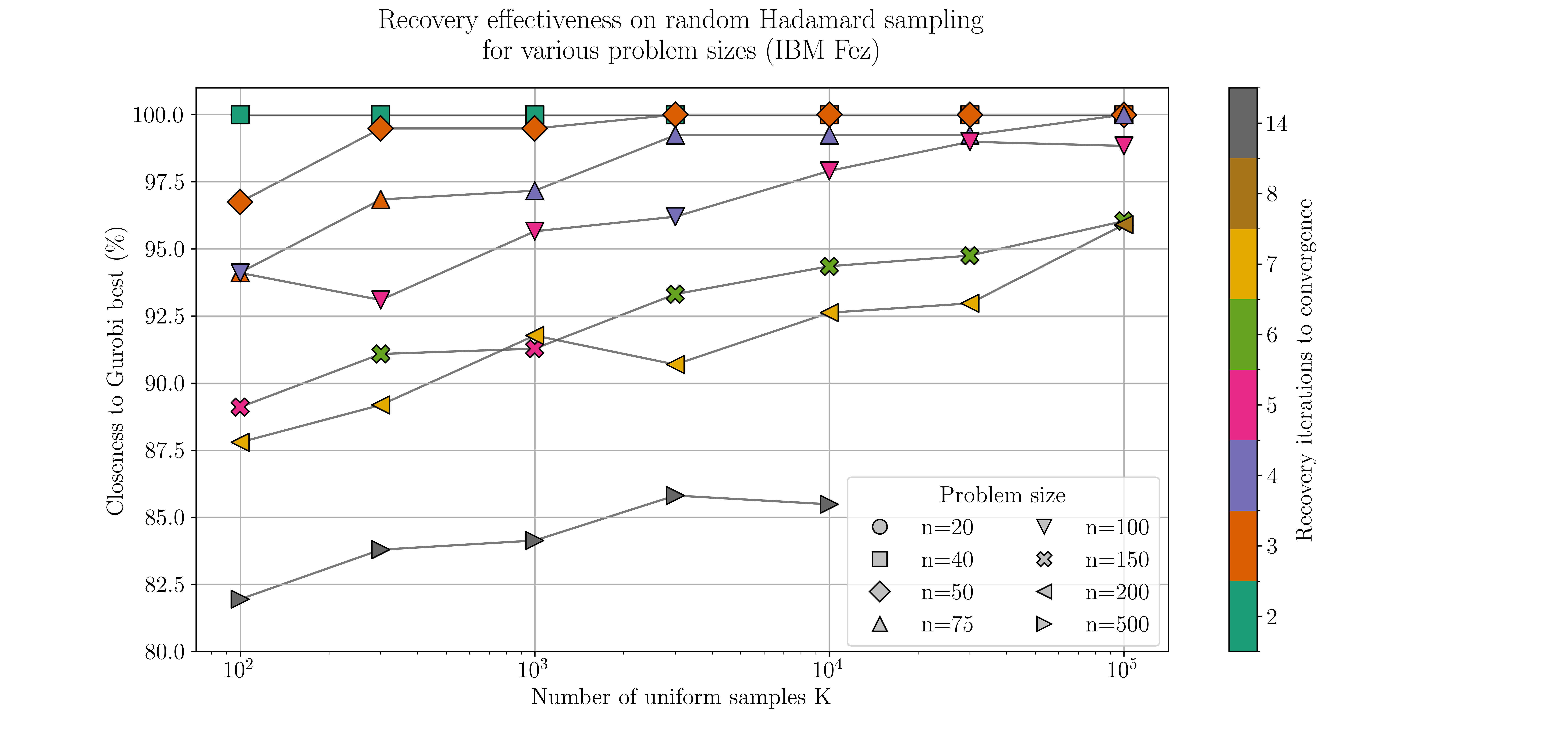}
    \caption{Effectiveness of the unconstrained recovery algorithm on K-randomly-sampled bitstrings from a uniform distribution obtained via Hadamard initialization on IBM Heron processors for $n \le 150$ and using classical sampling for $n=200$ and $500$. Marker shape indicates the problem size $n$, and marker color denotes the number of recovery iterations to convergence. As problem size increases, the recovered solutions tend to be farther from the Gurobi-best solution and require more recovery steps. Increasing the number of K‑random samples tends to improve solution quality, reducing the closeness gap. At $n=500$, prohibitive runtimes necessitated a cap of 14 recovery iterations, and results with $N_{\mathrm{shots}} > 10^4$ are omitted.}
    \label{fig: recovery_uniform}
\end{figure*}

\begin{figure*}
    \centering
    \includegraphics[trim={0cm 0cm 0cm 1.25cm}, clip,width=0.9\linewidth]{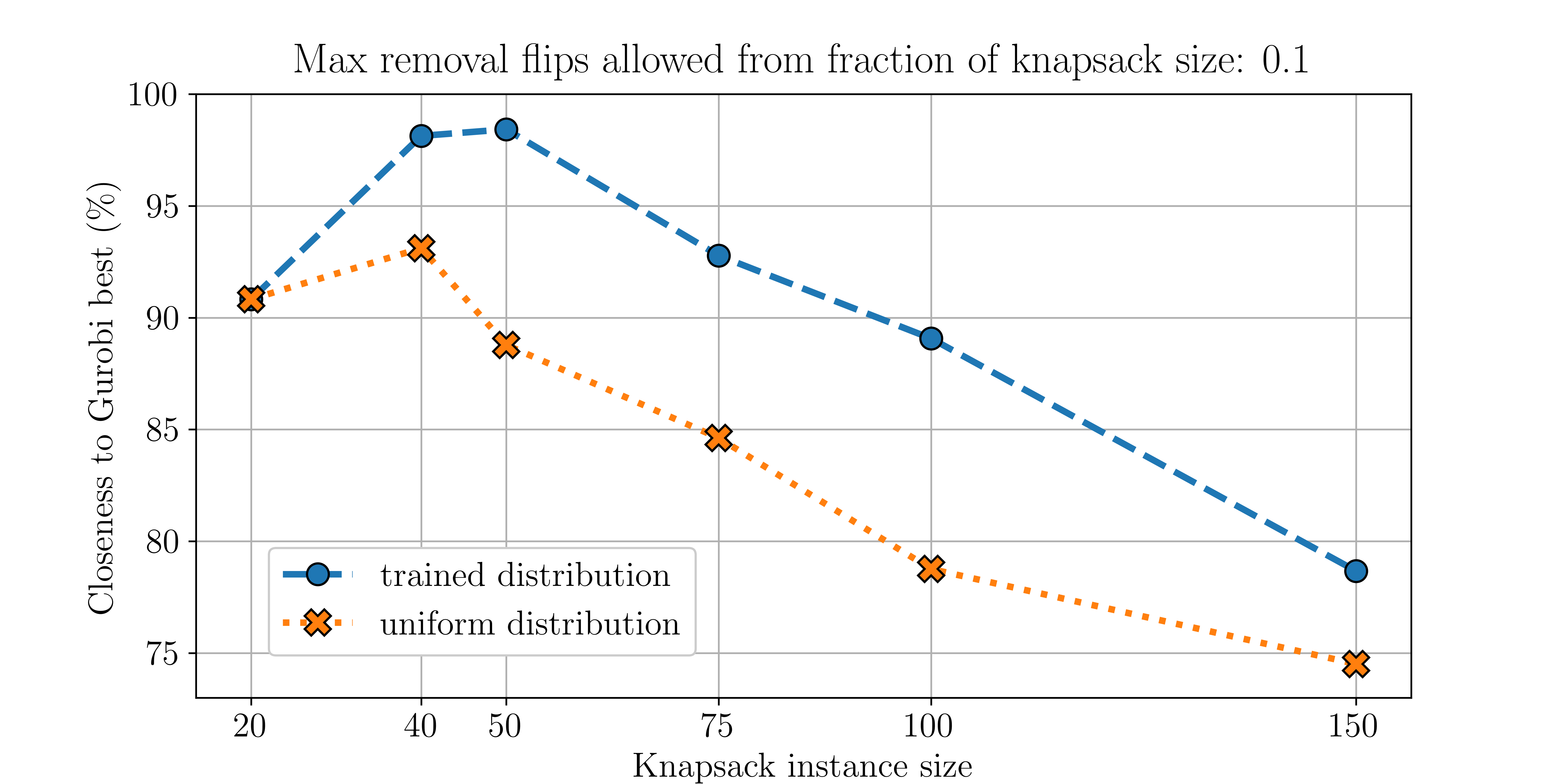}
    \caption{Closeness to the best Gurobi solution for knapsack instances of varying sizes under a recovery scheme with a limited removal budget equal to 10\% of the instance size. Recovery initialized from a trained quantum circuit (blue) consistently yields solutions closer to the Gurobi best than recovery initialized from a uniform distribution (orange). The performance gap remains relatively constant as the knapsack size increases, indicating that the trained distribution provides a higher-quality starting point for scalable instances when under recovery budget constraints.}
    \label{fig: recovery_limited}
\end{figure*}

The recovery procedure runtime is primarily dominated by feasibility checks and the number of distinct bitstrings processed.  Each feasibility evaluation requires matrix operations over $M \times n$ problem parameters and therefore costs $\mathcal{O}(Mn)$. For each recovery iteration, the algorithm processes a total of $\mathcal{S}$ unique bitstrings and applies the recovery routine to each of them. The routine performs up to $F_{\mathrm{rem}}$ specified feasibility-restoration steps, where each step evaluates the current bitstring together with up to $F_{\mathrm{add}}$ specified candidate flips. This yields at most $((F_{\mathrm{rem}} +F_{\mathrm{add}})n)$ feasibility evaluations per repaired sample, and the resulting worst-case per-iteration runtime scales as 
\begin{equation}
    \mathcal{O} \bigg( |\mathcal{S}|(F_{\mathrm{rem}}+F_{\mathrm{add}})Mn\bigg)
\end{equation}
When $F_{\mathrm{rem}}$, $F_{\mathrm{add}}$, and $M$ are treated as constants, this reduces to $\mathcal{O}(|\mathcal{S}|n)$, which is linear in both the number of processed bitstrings and problem size. 

In practice, finite sampling budgets further reduce the scaling to be linear in $n$, and \autoref{fig: recovery_uniform} details the dependence of the recovery procedure on sampling resources. An ablation study was performed in which inputs to the recovery scheme were samples drawn from a random uniform distribution. The total number of samples was varied between $10^2$ and $10^5$, obtained either classically for $n > 150$ or from Hadamard-initialized quantum states on IBM Heron processors for $n \le 150$. For each sample budget and problem size with $n < 500$, the recovery algorithm was executed until convergence, and both the number of recovery iterations required and the quality of the recovered solution relative to the Gurobi-best solution were recorded. 

As expected, reducing the number of samples significantly decreases the probability of observing high‑quality bitstrings near the Gurobi‑best solution in the raw measurement data. This effect becomes increasingly pronounced as the problem size grows: the number of configurations expands combinatorially, while the probability mass associated with constraint‑saturating, high‑profit solutions correspondingly concentrates on a vanishingly small fraction of the state space. With fewer shots, such configurations are therefore sampled more sparsely or may be entirely absent, requiring additional recovery iterations to reconstruct them through local repairs.

Consistent with this interpretation, larger problem instances and lower shot budgets exhibit a marked increase in the number of recovery steps required to reach convergence, along with a degradation in the best recovered objective value relative to the Gurobi-best. This is particularly pronounced when $n=500$, where excessive runtimes led to instantiating a cap on the number of recovery iterations to 14 and omission of results with $N_{\mathrm{shots}} > 10^4$. Conversely, increasing the shot budget improves sampling diversity, raising the likelihood that near‑optimal bitstrings are present in the initial pool and accelerating convergence of the recovery procedure. Taken together, these results isolate the role of finite‑shot sampling noise in governing recovery performance, disentangling sampling‑driven effects from problem‑size‑dependent computational hardness.

In addition, the performance of the recovery algorithm depends upon the sampling distribution upon which it acts. \autoref{fig: recovery_limited} examines the effectiveness of the recovery scheme constrained to a limited removal flip budget, where the quality of the initial distribution becomes the dominant factor governing final performance. In this regime, recovery initialized from samples obtained via the quantum algorithm and optimized circuit consistently produce solutions closer to the Gurobi best across all instance sizes when compared to uniform distribution starting points. This behavior reflects the fact that a restricted recovery budget cannot compensate for poor initial samples. In particular, when starting from a uniform distribution, the sampled bitstrings are typically far from high-quality, constraint-satisfying configurations, requiring large coordinated changes that exceed the allowed number of flips. In contrast, the trained circuit distribution concentrates probability mass closer to near-feasible regions of the solution space, so that only small local repairs are needed to reach high-quality solutions. As the instance size grows, this distinction remains relatively constant, highlighting that under realistic, budget-limited recovery, improvements in the learned distribution translate directly into improved end-to-end solution quality.

\section{Classical solvers} \label{app:classical_solvers}

\subsection{Problem Formulation}

All classical heuristic solvers use the deterministic form of the objective function \cite{GOYAL2010161}. For each constraint $m$, the left-hand side is computed as:
\begin{equation}
L_m(\bm{x}) = \sum_{i=1}^n \mu_i^{(m)} x_i + \Phi^{-1}(1-\epsilon) \sqrt{\sum_{i=1}^n (\sigma_i^{(m)})^2 x_i}
\end{equation}

The constraint violation is defined as $V_m(\bm{x}) = \max(0, L_m(\bm{x}) - C^{(m)})$, and the total violation across all constraints is $V(\bm{x}) = \sum_{m=1}^M V_m(\bm{x})$. The energy function to be maximized is then:
\begin{equation}
E(\bm{x}) = \bm{v}^\top \bm{x} - \lambda \cdot V(\bm{x}),
\end{equation}
where $\lambda = \sum_{i=1}^n v_i$ is the penalty multiplier, chosen as the sum of all item values to provide a principled balance between objective maximization and constraint satisfaction.
A solution is considered feasible if $L_m(\bm{x}) \leq C^{(m)}$ for all constraints $m \in \{1, \ldots, M\}$. Additionally, throughout this section, $\bm{x}^*$ denotes the best feasible solution encountered so far; infeasible solutions may guide the search via the energy function but are not reported as final solutions.



All metaheuristic solvers in the following sections initialize solutions uniformly at random, with each binary variable $x_i$ drawn independently from a Bernoulli distribution with probability 0.5. This ensures unbiased starting points across different runs and problem instances. For Parallel Tempering, each of the $K$ replicas is initialized independently. For the Genetic Algorithm, all $P$ individuals in the initial population are generated independently using this random initialization scheme.

All metaheuristics are subject to a wall-clock limit of 1800 seconds per
instance. In addition, parallel tempering and simulated annealing are capped
at 200{,}000 iterations, and the genetic algorithm at 2{,}000 generations with
a population of 100 (i.e., 200{,}000 offspring evaluations). Tabu search
differs in that each iteration evaluates up to 100 candidate neighbors rather
than a single proposal; its cap is therefore set to 32{,}000 iterations,
chosen such that its total number of loss-function evaluations
($\approx 3.2\times10^{6}$) matches that of parallel tempering, whose
8 replicas each evaluate two solutions per iteration.

\subsection{Gurobi (Exact Solver)}

See \autoref{sec:instances} for more about Gurobi.

\subsection{Parallel Tempering}

Parallel Tempering (PT) maintains multiple replicas of the solution at different temperatures, allowing efficient exploration of the solution space through both local moves and temperature swaps.

\begin{algorithm}
\caption{Parallel Tempering for Chance-Constrained knapsack}
\begin{algorithmic}[1]
\Require $n$ items, $M$ constraints, parameters $\bm{v}, \boldsymbol{\mu}, \boldsymbol{\sigma}, \bm{C}, \epsilon$
\Require $K$ replicas, temperatures $T_1 < T_2 < \cdots < T_K$, swap interval $S$
\State Initialize $K$ random solutions $\bm{x}^{(1)}, \ldots, \bm{x}^{(K)}$
\State $\bm{x}^* \gets \arg\max_k E(\bm{x}^{(k)})$ \Comment{Best solution}
\For{$t = 1$ to $T_{\max}$ or until time limit}
    \For{$k = 1$ to $K$} \Comment{Metropolis step for each replica}
        \State $\bm{x}' \gets \bm{x}^{(k)}$
        \State Flip random bit: $x'_i \gets 1 - x_i$ for random $i$
        \State $\Delta E \gets E(\bm{x}') - E(\bm{x}^{(k)})$
        \If{$\Delta E > 0$ or $\text{rand}() < \exp(\Delta E / T_k)$}
            \State $\bm{x}^{(k)} \gets \bm{x}'$ \Comment{Accept move}
        \EndIf
        \If{$E(\bm{x}^{(k)}) > E(\bm{x}^*)$ and $\bm{x}^{(k)}$ is feasible}
            \State $\bm{x}^* \gets \bm{x}^{(k)}$
        \EndIf
    \EndFor
    \If{$t \bmod S = 0$} \Comment{Attempt replica swaps}
        \For{$k = 1$ to $K-1$}
            \State $\Delta \gets (E(\bm{x}^{(k)}) - E(\bm{x}^{(k+1)})) \cdot (1/T_k - 1/T_{k+1})$
            \If{$\Delta > 0$ or $\text{rand}() < \exp(\Delta)$}
                \State Swap $\bm{x}^{(k)} \leftrightarrow \bm{x}^{(k+1)}$
            \EndIf
        \EndFor
    \EndIf
\EndFor
\State \Return $\bm{x}^*$
\end{algorithmic}
\end{algorithm}

\subsubsection{Implementation Details}

Our implementation uses $K = 8$ temperature replicas with temperatures geometrically spaced between $T_{\min} = 0.1$ and $T_{\max} = 10.0$. Each replica performs Metropolis-Hastings steps where a single random bit is flipped to generate a neighbor solution. The acceptance probability follows the standard Metropolis criterion: moves that increase energy are always accepted, while moves that decrease energy are accepted with probability $\exp(\Delta E / T_k)$ where $\Delta E$ is the energy difference and $T_k$ is the replica's temperature. Every $S = 100$ iterations, adjacent replicas attempt to swap configurations using the parallel tempering acceptance criterion, which depends on both the energy difference and temperature difference between replicas.

\subsection{Tabu Search}

Tabu Search (TS) uses a memory structure to avoid cycling and encourage exploration of new regions of the solution space.

\begin{algorithm}
\caption{Tabu Search for Chance-Constrained knapsack}
\begin{algorithmic}[1]
\Require $n$ items, $M$ constraints, parameters $\bm{v}, \boldsymbol{\mu}, \boldsymbol{\sigma}, \bm{C}, \epsilon$
\Require Tabu tenure $\tau$, max neighbors $N_{\max}$
\State Initialize random solution $\bm{x}$
\State $\bm{x}^* \gets \bm{x}$ \Comment{Best solution}
\State $\mathcal{T} \gets \emptyset$ \Comment{Tabu list (FIFO queue)}
\For{$t = 1$ to $T_{\max}$ or until time limit}
    \State $\mathcal{N} \gets \{\bm{x}' : \bm{x}'$ differs from $\bm{x}$ in one bit$\}$
    \If{$|\mathcal{N}| > N_{\max}$}
        \State Sample $N_{\max}$ neighbors randomly from $\mathcal{N}$
    \EndIf
    \State $\bm{x}_{\text{best}} \gets \arg\max_{\bm{x}' \in \mathcal{N}} E(\bm{x}')$
    \State $i_{\text{flip}} \gets$ index of flipped bit in $\bm{x}_{\text{best}}$
    \If{$i_{\text{flip}} \notin \mathcal{T}$ or $E(\bm{x}_{\text{best}}) > E(\bm{x}^*)$} \Comment{Aspiration}
        \State $\bm{x} \gets \bm{x}_{\text{best}}$
        \State Add $i_{\text{flip}}$ to $\mathcal{T}$ (remove oldest if $|\mathcal{T}| = \tau$)
        \If{$E(\bm{x}) > E(\bm{x}^*)$ and $\bm{x}$ is feasible}
            \State $\bm{x}^* \gets \bm{x}$
        \EndIf
    \EndIf
\EndFor
\State \Return $\bm{x}^*$
\end{algorithmic}
\end{algorithm}

\subsubsection{Implementation Details}

Our Tabu Search implementation maintains a tabu list with tenure $\tau = 20$ iterations, implemented as a first-in-first-out (FIFO) queue that stores the indices of recently flipped bits. At each iteration, we generate the neighborhood by considering single-bit flips if $n \leq N_{\max}$ from the current solution. When the neighborhood size exceeds $N_{\max} = 100$, we randomly sample $N_{\max}$ neighbors to maintain computational efficiency. The best neighbor is selected based on the energy function $E(\bm{x})$, and the corresponding bit flip is executed unless it is tabu. However, we employ an aspiration criterion that overrides the tabu status if the move leads to a solution better than the current best known solution.

\subsection{Simulated Annealing}

Simulated Annealing (SA) uses a temperature-based acceptance criterion with geometric cooling to gradually reduce exploration and focus on exploitation.

\begin{algorithm}
\caption{Simulated Annealing for Chance-Constrained knapsack}
\begin{algorithmic}[1]
\Require $n$ items, $M$ constraints, parameters $\bm{v}, \boldsymbol{\mu}, \boldsymbol{\sigma}, \bm{C}, \epsilon$
\Require Initial temp $T_0$, final temp $T_f$, cooling rate $\alpha$, iterations per temp $I$
\State Initialize random solution $\bm{x}$
\State $\bm{x}^* \gets \bm{x}$ \Comment{Best solution}
\State $\mathcal{T} \gets T_0$
\While{$\mathcal{T} > T_f$ and time limit not exceeded}
    \For{$i = 1$ to $I$}
        \State $\bm{x}' \gets \bm{x}$
        \State Flip random bit: $x'_i \gets 1 - x_i$ for random $i$
        \State $\Delta E \gets E(\bm{x}') - E(\bm{x})$
        \If{$\Delta E > 0$ or $\text{rand}() < \exp(\Delta E / \mathcal{T})$}
            \State $\bm{x} \gets \bm{x}'$ \Comment{Accept move}
        \EndIf
        \If{$E(\bm{x}) > E(\bm{x}^*)$ and $\bm{x}$ is feasible}
            \State $\bm{x}^* \gets \bm{x}$
        \EndIf
    \EndFor
    \State $\mathcal{T} \gets \alpha \cdot \mathcal{T}$ \Comment{Geometric cooling}
\EndWhile
\State \Return $\bm{x}^*$
\end{algorithmic}
\end{algorithm}

\subsubsection{Implementation Details}

Our Simulated Annealing implementation begins with an initial temperature of $T_0 = 10.0$ and cools geometrically with rate $\alpha = 0.95$ until reaching a final temperature of $T_f = 0.01$. At each temperature level, we perform $I = 1000$ iterations where a single random bit is flipped to generate a neighbor solution. The acceptance criterion follows the Metropolis rule: energy-increasing moves are always accepted, while energy-decreasing moves are accepted with probability $\exp(\Delta E / \mathcal{T})$ where $\Delta E$ is the energy change and $\mathcal{T}$ is the current temperature. After completing all iterations at a given temperature, we update the temperature as $T_{\text{new}} = \alpha \cdot T_{\text{old}}$.

\subsection{Genetic Algorithm}

The Genetic Algorithm (GA) maintains a population of solutions and evolves them using selection, crossover, and mutation operators inspired by natural evolution.

\begin{algorithm}
\caption{Genetic Algorithm for Chance-Constrained knapsack}
\begin{algorithmic}[1]
\Require $n$ items, $M$ constraints, parameters $\bm{v}, \boldsymbol{\mu}, \boldsymbol{\sigma}, \bm{C}, \epsilon$
\Require Population size $P$, crossover rate $p_c$, mutation rate $p_m$, tournament size $k$, elitism $e$
\State Initialize population $\mathcal{P} = \{\bm{x}_1, \ldots, \bm{x}_P\}$ randomly
\State $\bm{x}^* \gets \arg\max_{\bm{x} \in \mathcal{P}} E(\bm{x})$ \Comment{Best solution}
\For{$g = 1$ to $G_{\max}$ or until time limit}
    \State Evaluate fitness: $f_i \gets E(\bm{x}_i)$ for all $\bm{x}_i \in \mathcal{P}$
    \State Sort $\mathcal{P}$ by fitness (descending)
    \State $\mathcal{P}_{\text{new}} \gets \{\bm{x}_1, \ldots, \bm{x}_e\}$ \Comment{Elitism}
    \While{$|\mathcal{P}_{\text{new}}| < P$}
        \State $\bm{p}_1 \gets$ TournamentSelection$(\mathcal{P}, k)$
        \State $\bm{p}_2 \gets$ TournamentSelection$(\mathcal{P}, k)$
        \If{$\text{rand}() < p_c$}
            \State $(\bm{c}_1, \bm{c}_2) \gets$ UniformCrossover$(\bm{p}_1, \bm{p}_2)$
        \Else
            \State $(\bm{c}_1, \bm{c}_2) \gets (\bm{p}_1, \bm{p}_2)$
        \EndIf
        \State $\bm{c}_1 \gets$ Mutate$(\bm{c}_1, p_m)$
        \State $\bm{c}_2 \gets$ Mutate$(\bm{c}_2, p_m)$
        \State Add $\bm{c}_1, \bm{c}_2$ to $\mathcal{P}_{\text{new}}$
    \EndWhile
    \State $\mathcal{P} \gets \mathcal{P}_{\text{new}}$
    \If{$\max_{\bm{x} \in \mathcal{P}} E(\bm{x}) > E(\bm{x}^*)$ and feasible}
        \State $\bm{x}^* \gets \arg\max_{\bm{x} \in \mathcal{P}} E(\bm{x})$
    \EndIf
\EndFor
\State \Return $\bm{x}^*$
\end{algorithmic}
\end{algorithm}

\subsubsection{Genetic Operators}

Tournament selection chooses a parent by randomly selecting $k$ individuals from the population and returning the one with the highest fitness. Uniform crossover creates two offspring from two parents by independently choosing each bit position from either parent with equal probability, occurring with probability $p_c$. Bit-flip mutation independently flips each bit with probability $p_m$, introducing random variation to maintain population diversity.

\subsubsection{Implementation Details}

Our Genetic Algorithm maintains a population of $P = 100$ individuals, each representing a binary solution vector. At each generation, we first evaluate the fitness of all individuals using the energy function $E(\bm{x})$. We preserve the $e = 2$ best individuals through elitism, directly copying them to the next generation. The remaining population is generated through tournament selection with tournament size $k = 3$, followed by uniform crossover with probability $p_c = 0.8$ and bit-flip mutation with per-bit probability $p_m = 0.01$.

\end{document}